\documentclass[a4paper,UKenglish,cleveref,autoref,thm-restate,authorcolumns]{lipics-v2019}

\nolinenumbers

\usepackage{array}
\usepackage{picins}
\usepackage{multicol}
\usepackage{graphbox}
\usepackage{multirow}
\usepackage[T1]{fontenc}
\usepackage{mathtools}
\usepackage{amsfonts}
\usepackage{amssymb}
\usepackage{amsmath}
\usepackage{fixltx2e}
\usepackage{stmaryrd}
\usepackage{todonotes}
\usepackage[inline]{enumitem}
\usepackage{tikz-cd}
\usepackage{url}
\usepackage[labelfont=bf]{caption}
\usepackage{subcaption}
\usepackage{amsthm}
\usepackage{cleveref}
\usepackage{xypic}
\usepackage{tikz}
\usepackage{scalefnt}

\usepackage{wrapfig}

\xyoption{rotate}
\xyoption{pdf}

\newcommand{\Set}{\mathbf{Set}}

\newcommand{\Alg}[1][\Sigma,E]{\mathbf{Alg}(#1)}
\newcommand{\cat}{\mathbf{C}}

\newcommand{\EM}{\mathcal{E\!M}}

\newcommand{\po}{\mathscr{P}}
\newcommand{\pof}{\mathscr{P}_f}
\newcommand{\Id}{\mathrm{Id}}

\newcommand{\di}{\mathscr{D}}
\newcommand{\poly}[1][\Sigma]{\mathsf{H}_{#1}}
\newcommand{\free}[1][\Sigma]{\mathsf{F}_{#1}}

\newcommand{\Th}{\mathbb{T}}

\newcommand{\algb}[1][A]{\mathcal{#1}}
\newcommand{\swp}{\mathsf{swap}}

\newcommand{\unit}{\left[ 0,1\right]}

\newcommand{\id}{\mathrm{id}}

\newcommand*{\ccld}[1]{{\tikz[baseline=(X.base)]\node(X)[draw,shape=circle,inner sep=0]{\text{\scriptsize\strut$#1$}};}}
\newdir{:=}{{}}

\newcommand{\prepare}[1][\algb]{\delta_{#1}}
\newcommand{\evaluate}[1][\algb]{\gamma_{#1}}
\newcommand\te{\psi}

\newcommand{\prob}{\nu}

\newcommand{\ari}[1]{|{#1}|}

\newcommand{\semi}{\mathbb{S}}
\newcommand{\ms}{\mathbf{M_{\semi}}}
\newcommand{\mzt}{\mathbf{M_{\mathbb{Z}_2}}}

\newcommand{\Var}{\mathrm{Var}}
\newcommand{\Arg}{\mathrm{Arg}}

\makeatletter
\newcommand{\overleftrightsmallarrow}{\mathpalette{\overarrowsmall@\leftrightarrowfill@}}
\newcommand{\overrightsmallarrow}{\mathpalette{\overarrowsmall@\rightarrowfill@}}
\newcommand{\overleftsmallarrow}{\mathpalette{\overarrowsmall@\leftarrowfill@}}
\newcommand{\overarrowsmall@}[3]{%
  \vbox{%
    \ialign{%
      ##\crcr
      #1{\smaller@style{#2}}\crcr
      \noalign{\nointerlineskip}%
      $\m@th\hfil#2#3\hfil$\crcr
    }%
  }%
}
\def\smaller@style#1{%
  \ifx#1\displaystyle\scriptstyle\else
    \ifx#1\textstyle\scriptstyle\else
      \scriptscriptstyle
    \fi
  \fi
}
\makeatother
\newcommand{\mtr}[1]{\overleftrightsmallarrow{#1}}

\newcolumntype{P}[1]{>{\centering\arraybackslash}p{#1}}
\newcolumntype{M}[1]{>{\centering\arraybackslash}m{#1}}

\usepackage{booktabs}   
\usepackage{subcaption} 

\author{Louis Parlant}{University College London}{l.parlant@cs.ucl.ac.uk}{}{}
\author{Jurriaan Rot}{Radboud University}{jrot@cs.ru.nl}{}{}
\author{Alexandra Silva}{University College London}{alexandra.silva@ucl.ac.uk}{}{}
\author{Bas Westerbaan}{University College London}{bas@westerbaan.name}{}{}

\Copyright{Louis Parlant, Jurriaan Rot, Alexandra Silva and Bas Westerbaan}

\authorrunning{Parlant, Rot, Silva and Westerbaan}



\funding{%
This work was partially supported by the
ERC Starting Grant
ProFoundNet (grant code 679127), the
Marie Curie Fellowship CARBS
(grant code 795119), and the
EPSRC
   Standard Grant CLeVer (EP/S028641/1).}

\acknowledgements{
	We are grateful to Dan Marsden, Robin Piedeleu, Edmund Robinson and Maaike Zwart
	for inspiring discussions and suggestions. 
}

\relatedversion{
A full version, which includes proofs, 
is available at~\url{https://arxiv.org/abs/2001.06348}.
}

\EventEditors{Javier Esparza and Daniel Kr{\'a}l'}
\EventNoEds{2}
\EventLongTitle{45th International Symposium on Mathematical Foundations of Computer Science (MFCS 2020)}
\EventShortTitle{MFCS 2020}
\EventAcronym{MFCS}
\EventYear{2020}
\EventDate{August 24--28, 2020}
\EventLocation{Prague, Czech Republic}
\EventLogo{}
\SeriesVolume{170}
\ArticleNo{75}

\begin{document}

\title{Preservation of Equations by Monoidal Monads}         

\maketitle

\begin{abstract}
If a monad~$T$ is monoidal, then operations on a set~$X$ can be lifted
canonically to operations on~$TX$. In this paper we study structural
properties under which~$T$ preserves equations between those
operations. It has already been shown that any monoidal monad preserves
linear equations; \emph{affine} monads preserve \emph{drop} equations
(where some variable appears only on one side, such as~$x\cdot y = y$)
and \emph{relevant} monads preserve \emph{dup} equations (where some
variable is duplicated, such as~$x \cdot x = x$). We start the paper by
showing a converse: if the monad at hand preserves a drop equation, then
it must be affine. From this, we show that the problem whether a given
(drop) equation is preserved is undecidable. A converse for relevance
turns out to be more subtle: preservation of certain 
dup equations implies a weaker notion which we call $n$-relevance. Finally,
we identify a subclass of equations such that their preservation is
equivalent to relevance.
\end{abstract}

\section{Introduction}

Monads are fundamental structures in programming language semantics as they encapsulate many common side-effects such as non-determinism, exceptions, or randomisation. Their structure has been studied not only from an operational point of view but also from a categorical and algebraic perspective. One question that has attracted much attention is that of monad composition: given two monads $T$ and $S$, is the composition $TS$ again a monad? The answer to the question in full generality is subtle and examples have shown that even in simple cases the correct answer might be surprisingly tricky to find and prove. 
This subtlety is illustrated, for instance, by Klin and Salamanca's result that the powerset monad does not compose with itself~\cite{KlinPP}, invalidating repeated claims to the contrary in the literature. 


Monad composition can be viewed in different ways. On the one hand, a sufficient condition is given by the existence of the categorical notion of a distributive law between the monads. 
 On the other hand, if one takes into account the algebraic structure of the inner monad, which can be presented in a traditional operations plus equations fashion, one can turn the original question into a preservation question:  does the outer monad preserve all operations and equations of the inner algebra? More generally, starting from a set $A$ with some additional algebraic structure, the question arises naturally: is this structure preserved by the application of a monad? For example, consider a set $A$ that has the structure of a group with binary operation $\cdot$ and unit $1$. As we apply the \emph{powerset monad} $\po$ the set $\po A$ of subsets of $A$, is $\po A$ is again a group? Can $\cdot$ be interpreted as a binary operation on elements of $\po A$? Which subset do we identify with the constant $1$? In a nutshell, does our monad preserve algebraic features, i.e., the operations and equations defining the structure of $A$?

This question has already been studied in the late 50s, albeit not in categorical terms. In~\cite{gautam1957validity}, Gautam introduces the notion of \emph{complex algebra}---the transformation of a $\Sigma$-algebra with carrier $A$ into a $\Sigma$-algebra on $\po A$,
where $\Sigma$ is an arbitrary signature. 
Gautam gives a range of positive and negative results for equation preservation. In particular, he shows that commutativity of a binary operation ($x \cdot y = y \cdot x$) and unitality ($x \cdot 1= x$) are unconditionally preserved by $\po$. These are examples of \emph{linear} equations, in which each variable appears exactly once on each side.  Negative results are given for non-linear equations. First, if variables appear more than once (we call these \emph{dup} equations; e.g., $x \cdot x = x$), this is not preserved by $\po$.  Second, if a variable appears on one side of the equation only (we call these \emph{drop} equations, for instance $x \cdot 0 = 0$), then the powerset does not preserve it either.

In this paper, we examine the question of equation preservation at the general level 
of a \emph{monoidal monad} $T$.
Monoidal monads give a canonical lifting to $\Sigma$-algebras, of which Gautam's complex
algebra construction for $\po$ is a special case. 
We provide a comprehensive characterisation of monad classes and equations that are preserved inside a certain class. 

Part of this question has been studied in the literature: in \cite{manes2007monad}, Manes and Mulry show that for a monad to preserve linear equations, it suffices to have a \emph{symmetric monoidal} structure. This argument is further developed in \cite{layers}, where the authors give sufficient conditions on $T$ for preserving the types of equations outlined by Gautam.  In particular, it was shown that so-called  \emph{relevant} monads preserve {\it dup} equations, and \emph{affine} monads preserve \emph{drop} equations. It remained open whether relevance or affineness were necessary conditions for preservation. We now settle this question and provide an extensive characterisation of equations preserved by classes of monoidal monads. 

We start with preliminaries on monoidal categories and monads (Section~\ref{sec:prelim}), and continue with a technical section recalling how monoidal monads offer a canonical lifting of all algebraic signatures (Section~\ref{sec:preserve}). We then present the main contributions of the present paper: 
\begin{enumerate}[leftmargin=*]
\item We prove a monoidal monad preserves strict-drop equations if and only if it is affine (Section~\ref{sec:affine}). 
\item We characterise a large class of \textit{dup} equations, for which preservation is equivalent to the monad being relevant. We then prove that for a restricted class of \textit{dup} equations preservation requires a weaker version of relevance, which we call $n$-relevance (Section~\ref{sec:dup}).
\item Orthogonally, we prove a more algorithmic result: given a monad and an equation, we show that the general problem of preservation is undecidable (Section~\ref{sec:affine}).
\end{enumerate}
A summary of the classes of equations used to derive the preservation results and concrete examples of monads and the class in which they fall appear in Figure~\ref{monadtable}.

We remark that the provided necessary conditions on preservation of equations can be used in contrapositive form to show that certain monads do not preserve certain equations. 
 For instance, from Theorem~\ref{thm:drop-affine}, the main result of Section~\ref{sec:affine}, we know that if a monad is not affine, it does not preserve any drop equations in general. This generalises Gautam's result for $\po$ and provides a range of other examples that do not preserve drop equations: the maybe monad $X + 1$, the monad $M \times X$ for $M$ a non-trivial monoid, and the multiset monad $\ms$, for $\semi$ a non-trivial semiring.
Similarly, the results in Section~\ref{sec:dup} allow us to show that certain (in fact, many) monads
do not preserve dup equations, just by showing that they are not relevant. 

\section{Preliminaries}
\label{sec:prelim}

We recall basic notions related to monoidal categories, monads and algebras for a functor. 

\subparagraph*{Cartesian monoidal categories. } 
%
A \emph{monoidal category} consists of a category $\cat$ equipped with:
a bifunctor $\otimes \colon \cat \times \cat \to \cat$,
an object $I$ called \emph{unit} or \emph{identity},
a natural isomorphism $\alpha_{X,Y,Z}\colon  (X\otimes Y)\otimes Z \to X\otimes (Y\otimes Z)$ called \emph{associator},
a natural isomorphism $\rho_{X}\colon  X \otimes I \to X$ called \emph{right unitor}, and
a natural isomorphism $\rho'_{X}\colon  I \otimes X \to X$ called \emph{left unitor},
where $\alpha, \rho$ and $\rho'$ are subject to coherence conditions~\cite{mac2013categories}. 
A \emph{Cartesian monoidal category} is a monoidal category whose monoidal structure is given by product (denoted by $\times$) 
and a terminal object (denoted by $1$).
A category with finite products forms a Cartesian monoidal category, by making a choice of products for each pair of objects. 
Any Cartesian monoidal category $\cat$ is symmetric monoidal, witnessed by a natural 
isomorphism denoted by $\swp_{X,Y} \colon X \times Y \rightarrow Y \times X$. For a product $\prod_{i \in I} X_i$, we
denote the projections by $\pi_j \colon \prod_{i \in I} X_i \rightarrow X_j$. Given an object $X$, we define
the \emph{diagonal} $\Delta_X$ by pairing: $\Delta_X = \langle id_X, id_X \rangle \colon X \rightarrow X \times X$. 
Further, for a functor $T \colon \cat \rightarrow \cat$ we let $\chi_{X,Y} = \langle T\pi_1, T\pi_2 \rangle 
\colon T(X \times Y) \rightarrow TX \times TY$. 
Throughout this paper, we will mainly focus on the Cartesian monoidal category  
$\Set$ of sets and functions. 


%
%
%

\subparagraph*{Monads} 
%
A \emph{monad} on a category $\cat$ is a triple $(T,\eta,\mu)$ consisting of an endofunctor
$T \colon \cat \rightarrow \cat$, a natural transformation
$\eta \colon \Id \Rightarrow T$ called \emph{unit} and a natural transformation 
$\mu \colon TT \Rightarrow T$ called \emph{multiplication}, 
such that $\mu \circ T \mu = \mu \circ \mu T$ and $\mu \circ \eta T = \id = \mu \circ T \eta$. 
%
%
%
\begin{example} [Monads] \label{ex:monads}
\begin{enumerate}[leftmargin=*]
\item 
The \emph{powerset monad} is given by the functor $\po \colon \Set \to \Set$, 
which maps a set $X$ to its set of subsets, together with unit $\eta$ mapping $x$ to $\{x\}$,
and $\mu$ given by union. 
This monad restricts to its finitary version with the functor $\pof \colon \Set \to \Set$ mapping a set to its finite subsets, and similarly to the \emph{non-empty powerset} $\po^+$.
%

%
%

\item For a semiring $\semi$, let $\ms$ be the \emph{generalised multiset} monad defined as follows.
A multiset~$\xi \in \ms X$
    is a map~$X \to \semi$
    with finite support,
    also written as the formal sum~$\sum_x \xi(x) x$.
For $f\colon  X \to Y$
    we define~$(\ms f)(\xi)(y) = \sum_{x;\ f(x)=y} \xi(x)$.%
 The unit is defined by
$\eta(x) = 1\cdot x$
    and the multiplication is given
    by~$\mu(\delta)(f) = \sum_x \delta(f) f(x) $,
viz.~$\mu(\sum_i s_i \sum_j t_{ij} x_{ij}) = \sum_{i,j} s_i t_{ij} x_{ij}$.

%
%
%

\item The \emph{distribution monad} is given by $\di \colon \Set \rightarrow \Set$, $\di X=\{\prob \colon X\to\unit \mid \sum_{x\in X}\prob(x)=1, \prob\text{ with finite support}\}.$ 
The action on morphisms, unit, and multiplication are as in~$\ms$.
\item \label{ex:mtimesmonad}
For any commutative monoid~$M$,
    the functor~$X \mapsto M \times X$
    is a monad
    with~$\eta (x) = (1,x)$
    and~$\mu(v,(w,x)) = (vw, x)$.

\item For a set $A$, $X \to X^A$ forms a monad. The unit is defined by $\eta_X(x)=(a \mapsto x)$, and the multiplication $\mu_X$ maps $s \in (X^A)^A$ to $(a \mapsto s(a)(a))$.

\end{enumerate}
\end{example}

\subparagraph*{Monoidal monads} The notion of \emph{monoidal monad} captures a well-behaved interaction between
a monad and the monoidal structure of the underlying category. 

Let $\cat$ be a Cartesian monoidal category. 
A \emph{(lax) monoidal functor} is an endofunctor $F\colon \cat\to\cat$ together with natural transformations $\psi_{X,Y}\colon  FX\times FY\to F(X\times Y)$ and $\psi^0\colon  1\to F1$ satisfying certain laws. 
A monoidal functor is called \emph{symmetric} if
$\psi_{Y,X} \circ \swp_{FX,FY} = F\swp_{X,Y} \circ \psi_{X,Y}$ for all $X,Y$. 
%
A monad $(T,\eta,\mu)$ on $\cat$ is called \emph{monoidal}
if $T$ is monoidal, the unit $\eta$ and multiplication $\mu$ are monoidal natural transformations, and the associated natural transformation 
$\psi_{X,Y}\colon  TX\otimes TY\to T(X\otimes Y)$ satisfies $\psi^0=\eta_1$.
Since the underlying category $\cat$ is symmetric monoidal, 
monoidal monads
are equivalent to commutative monads~\cite{kock,kock1970monads}, see also~\cite{jacobs1994semantics}. It follows 
that any monoidal monad on a symmetric monoidal category is symmetric.

For any $n\geq 2$, we define $\te^n\colon  TX_1 \times \dots \times TX_n \to T(X_1 \times \dots \times X_n)$ as the $n$-ary version of $\te$, associatively constructed from the binary version.

\begin{example}\label{ex:monoidal} All the monads from Example~\ref{ex:monads} are monoidal, by defining $\psi$ as follows: 
\begin{enumerate}[leftmargin=*]
\item $\te_{X,Y} \colon \po X \times \po Y \to \po (X \times Y)$ is given by Cartesian product:
$\te_{X,Y}(U,V) = U \times V$.
\item $\te_{X,Y}  \colon {\ms X} \mathop\times {\ms Y} \to \ms {(X \times Y)}$ is
given by $\te_{X,Y}(\xi_1, \xi_2)=  \lambda (x,y) . \xi_1 (x) \cdot \xi_2(y)$. 
\item $\te_{X,Y} \colon \di X \times \di Y \to \di (X \times Y)$ is given by
$\te_{X,Y}(\nu_1, \nu_2) = \lambda (x,y). \nu_1(x) \cdot \nu_2(y)$. 
\item     $\te_{X,Y}\colon (M\times X) \times (M\times Y) \to M \times (X\times Y)$ is given by $\te_{X,Y}((v,x), (w,y)) = (vw, (x,y))$.
\item $\te_{X,Y}\colon X^A \times Y^A \to (X\times Y)^A$ is defined pointwise as $\te_{X,Y}(f,g)(a) = (f(a),g(a))$.
\end{enumerate}
\end{example}
\subparagraph*{Algebraic Constructs.}
%
%
 A \emph{signature} $\Sigma$ is a set of operation symbols, together with
a natural number $\ari{\sigma}$ for each $\sigma \in \Sigma$, called the \emph{arity} of $\sigma$. 
%
%
For $X$ a set, the set of \emph{$\Sigma$-terms} over $X$ is 
the least set $\Sigma^* X$ such that 
$X \subseteq \Sigma^* X$ and, 
if $t_1, \ldots, t_{\ari{\sigma}} \in \Sigma^* X$ for some $\sigma \in \Sigma$ then $\sigma(t_1, \dots t_n) \in \Sigma^* X$. 
We write $\Var(t)$ for the set of variables appearing in the term $t$.
%

An \emph{algebraic theory} $\Th$ is a triple $(\Sigma, V, E)$ where
$\Sigma$ is a signature,
$V$ is a set of variables, and 
$E \subseteq \Sigma^* V \times \Sigma^* V$ is a relation. We refer to elements of $E$ as \emph{equations} or \emph{axioms} of $\Th$,
and denote an equation $(u,v)\in E$ by $u=v$. 
When two $\Sigma$-terms $t_1,t_2$ can be proved equal using equational logic and the axioms of $\Th$, we write $t_1=t_2$.
More precisely, $t_1=t_2$ if $t_1$ and $t_2$ are related by the least congruence containing $E$ which is
also closed under substitution.

For instance, the theory of monoids has a signature containing one constant $1$ and a binary symbol $\cdot$ and the axioms of associativity and unit:
$x \cdot 1 =  x = 1 \cdot x$ and 
$x \cdot ( y \cdot z ) =  (x \cdot y) \cdot z$.
For a signature $\Sigma$ , a \emph{$\Sigma$-algebra} $\mathcal{A}$ consists of a carrier set $A$ and, for each symbol $\sigma \in \Sigma$ of arity $\ari{\sigma}$, a morphism $\sigma_\mathcal{A}\colon  A^{\ari{\sigma}} \to A$. Given a $\Sigma$-algebra $\mathcal{A}$ and a map $f \colon V \rightarrow A$ to its carrier, 
we inductively define $f^\sharp \colon \Sigma^* V \rightarrow A$ by $f^\sharp(x) = f(x)$ and 
$f^\sharp(\sigma(t_1, \ldots, t_n)) = \sigma_\mathcal{A}(f^\sharp(t_1), \ldots, f^\sharp(t_n))$. 
Given an equation $t_1 = t_2$ with $t_1, t_2 \in \Sigma^*V$ we say 
$\mathcal{A}$ \emph{satisfies} $t_1 = t_2$, denoted by $\mathcal{A} \models t_1 = t_2$,
if $f^\sharp(t_1) = f^\sharp(t_2)$ for every map $f \colon V \rightarrow A$. 
This extends to sets of equations $E \subseteq \Sigma^* V \times \Sigma^*V$ 
by $\mathcal{A} \models E$ iff $\mathcal{A} \models t_1 = t_2$ for all $(t_1, t_2) \in E$. 


Categorically speaking, a signature $\Sigma$ gives rise to a polynomial functor $\poly \colon \Set \rightarrow \Set$,
defined by $\poly X = \coprod_{\sigma \in \Sigma} X^{\ari{\sigma}}$. 
A $\Sigma$-algebra as defined above is then precisely an algebra for $\poly$, i.e., a set $A$ together with a
map $\poly A \rightarrow A$.
The category of $\Sigma$-algebras and $\Sigma$-algebra morphisms is denoted $\Alg[\Sigma]$. In particular, the set $\free X$ of free $\Sigma$-terms on $X$ is a $\Sigma$-algebra, and $\free$ forms a functor.
For a set of equations $E \subseteq \Sigma^* V \times \Sigma^*V$, we
denote by $\Alg[\Sigma,E]$ the full subcategory of $\Alg[\Sigma]$ consisting of those algebras satisfying $E$.

\section{Preserving Operations and Equations}\label{sec:preserve}

In this section we explain how to lift operations,
    which allows us to define preservation of equations.
We conclude with a technical reformulation of equations and preservation that will be useful later.
In this section, we assume $T$ is a monoidal monad on~$\Set$.

\subparagraph*{Lifting operations.} 
Let $\Sigma$ be a signature. Given a $\Sigma$-algebra $\algb$
with carrier $A$, we define a $\Sigma$-algebra $\widehat{T}\algb$ on $TA$:
for each operator $\sigma \in \Sigma$, we set
$
\sigma_{\widehat{T}\algb}
\ \equiv \
\bigl(
\xymatrix{
	(TA)^{\ari{\sigma}} \ar[r]^-{\psi^{\ari{\sigma}}}
		& TA^{\ari{\sigma}} \ar[r]^-{T\sigma_{\algb}} 
        & \displaystyle TA
}\bigr)
.
$
This gives a lifting $\widehat{T} \colon \Alg[\Sigma] \rightarrow \Alg[\Sigma]$ of $T$. In fact, this is a lifting of the \emph{monad} $T$,
as it arises from a canonical distributive law of $H_\Sigma$ over the monad $T$
from~\cite{sokolova2007generic} (cf.~\cite{parlant2020preservation}). 

\begin{example}\label{poliftsign}
Consider the powerset monad and an algebra~$\algb$ with carrier $A$ and
a binary operation $\cdot$.
The lifted operation is given by
$
    U \mathrel{\cdot_{\widehat{\po}\algb}} V
    \ \equiv \ \{u \cdot v;\  u\in U, v\in V \}
$.
Liftings for other monads are easily obtained from Example~\ref{ex:monoidal}.
\end{example}

\subparagraph*{Preserving equations.} 
The lifting of algebraic operations allows interpretation of equations after application of $T$: if $t_1=t_2$ holds 
on a $\Sigma$-algebra $\algb$, we can interpret $t_1$ and $t_2$ as terms of $\widehat{T} \algb$ and verify equality. 
This leads to the central notion of \emph{preservation} of equations.

\begin{definition}\label{def:equationpreserved}
Let $\Sigma$ be a signature, $V$ a set of variables, and $t_1,t_2 \in \Sigma^* V$. We say that $T$ 
\emph{preserves the equation $t_1=t_2$} if for every $\Sigma$-algebra $\algb$ we have
    $\widehat{T} \algb \models t_1=t_2 $ provided~$\algb \models t_1=t_2$.
A set of equations is preserved if each one of them is.
\end{definition}
Equivalently, $T$ preserves $E$ if $\widehat{T}$ restricts to $\Alg[\Sigma,E]$. 
In particular, if $S$ is a monad such that $\Alg[\Sigma,E] \cong \EM(S)$,
then $T$ preserving $E$ implies that $TS$ is again a monad. 

\begin{example}
Does the powerset monad preserve the equation of commutativity? Let us consider an algebra $\algb$ with carrier $A$ featuring a commutative operation $+$. Recalling the lifting defined in example~\ref{poliftsign}, we can verify that for $U,V\in \po A$ we have: 
$U +_{\widehat{\po}\algb} V \ =\  \{ u + v \mid u \in U, v \in V\}
 \ =\  \{ v + u \mid v \in V, u \in U\} 
\ =\  V +_{\widehat{\po}\algb} U$.
Therefore $\po$ preserves commutativity. 
\end{example}

\begin{example}
Consider now an algebra $\algb$ with carrier $A$ and an idempotent operation $\cdot$ and the monad $\di$ of probability distributions. Let $a,b \in A$, such that $a \cdot b \neq a$, $a \cdot b \neq b$ and $a \cdot b \neq b \cdot a$. 
Let $\prob \in \di \algb$ be the distribution on $A$ such that $\prob(a)=\prob(b)=0.5$. 
Note that, for all $u,v \in A$, we have $\nu(u) \cdot \nu(v) \neq 0$ iff $u,v \in \{a,b\}$. Now we compute:
$(\prob \cdot_{\widehat{\di}\algb} \prob)(a \cdot b) 
\ =\  \sum \{ \nu(u) \cdot \nu(v) \mid u,v \in A \text{ s.t. }u \cdot v = a \cdot b\}
\ =\ \nu(a) \cdot \nu(b)\ =\ 0.25 \ \neq\  0 \ =\  \nu (a \cdot b)$.
Thus we have $\nu \neq \nu \cdot_{\widehat{\di}\algb} \nu$, which means that idempotence is not preserved by $\di$.
\end{example}

In subsequent sections we will treat several classes of equations, which are preserved by different
types of monads. Here, we first recall a fundamental result about preservation: any monoidal monad
preserves linear equations.

\begin{definition}
An equation $t_1=t_2$ is called \emph{linear} when $\Var(t_1) = \Var(t_2)$, and every variable
occurs exactly once in $t_1$ and once in $t_2$. 
\end{definition}
For instance, the laws of associativity, commutativity and unit are linear. 
Note that if an equation is \emph{not} linear, then either there is a variable 
which occurs on one side but not the other (which we will refer to as a \emph{drop} equation, Definition~\ref{def:drop})
or there is a variable that occurs twice (referred to as a \emph{dup} equation, Definition~\ref{def:dup}).
Manes and Mulry showed that any monoidal monad $T\colon \Set \to \Set$ preserves linear equations~\cite{manes2007monad}. This generalises Gautam's result that the powerset monad preserves linear equations~\cite{gautam1957validity}. 

\begin{figure}[t]
      \subcaptionbox{Affineness and relevance of well-known monads\label{monadtable}}{
\begin{tabular}{@{} |@{\,} c@{\,} | @{\,} c@{\,}  | @{\,} c@{\,}  |@{}}
\hline
Monad & Affine & Relevant\\
\hline
$\po$ & $\times$ & $\times$\\
$\po^+$ & $\checkmark$ & $\times$\\
$\di$ & $\checkmark$ & $\times$\\
$X+1$ & $\times$ & $\checkmark$\\
$X^A$ & $\checkmark$ & $\checkmark$\\
$M \times X$ & $\checkmark$ iff $M$ trivial & $\checkmark$ iff $M$ idempotent\\
$\ms$ & $\checkmark$ iff $\semi$ trivial &  $\checkmark$ iff $\semi$ trivial \\
\hline
\end{tabular}
}\qquad
   \subcaptionbox{Classes of equations\label{venn}}{
\scalebox{.65}{\begin{picture}(200,100)
\put(0,0){\includegraphics[scale=0.175]{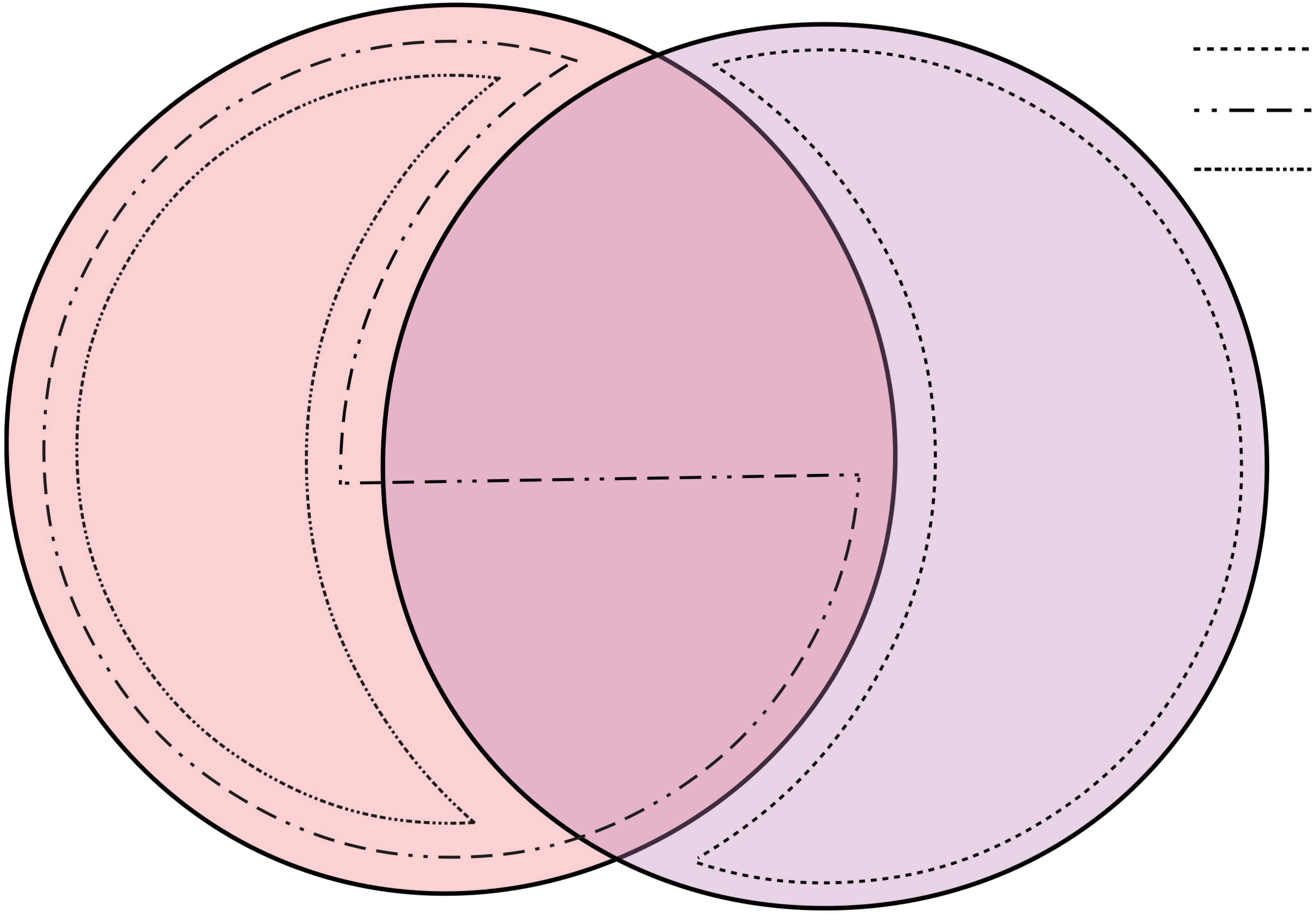}}
\put(50,163){drop}
\put(140,163){dup}
\put(18,100){$x\cdot y=y$}
\put(80,55){$x\cdot y=x\cdot x$}
\put(80,90){$x\cdot x=y\cdot y$}
\put(165,50){$x\cdot x=x$}
\put(240,150){strict-dup}
\put(240,140){one-drop}
\put(240,130){strict-drop}
\end{picture}}
}

\caption{}\label{figure}
\end{figure}

\section{Affine Monads and Drop Equations}\label{sec:affine}

In this section, we study preservation of \emph{drop} equations, which are non-linear
equations where at least one variable occurs on one side but not the other.
Preservation of 
such equations by a monoidal monad~$T$ will be related to a property of monoidal monads called \emph{affineness}.
\begin{definition}[Drop equations]\label{def:drop}
An equation $t_1=t_2$ is called
\begin{enumerate*}
	\item \emph{drop} when at least one variable appears in $t_1$ but not in $t_2$ (or conversely);
	\item \emph{one-drop} when a variable appears once in $t_1$ and does not appear in  $t_2$ (or conversely);
	\item \emph{strict-drop} when it is not linear and each variable of $V$ appears at most once in $t_1$ and $t_2$.
\end{enumerate*}
\end{definition}
The set of strict-drop equations is included in the set of one-drop equations,
and the latter in the set of drop equations. Both inclusions are strict. 
Strict-drop equations can equivalently be characterised as equations $t_1 = t_2$ where
neither $t_1$ nor $t_2$ contains duplicate variables, and at least one variable occurs on one side but not the other. 
The various classes of drop equations are pictured in Figure~\ref{venn}, which also contains dup equations (Definition~\ref{def:dup}).

For instance, the law of absorption $x \cdot  0 = 0$ is a strict-drop equation.
The equation $x\cdot (y\cdot y) = y\cdot y$ shows one occurrence of $x$ on the left side and none on the right side, therefore it is a one-drop equation.
It is not a strict-drop equation. The equation $x \cdot  x = y \cdot  y$ is drop but not one-drop. Finally, $x \cdot  x = x$ is not drop.

The property of $T$ that we focus on now is \emph{affineness}, introduced by Kock~\cite{kock1972bilinearity}; see also~\cite{jacobs1994semantics}.

\begin{definition}[\cite{kock1972bilinearity}]\label{definition:affine}
A monoidal monad $T$ on a Cartesian monoidal category~$\cat$
is called \emph{affine} if it has one of the following three
equivalent properties: $T1$ is final;  diagram $\spadesuit$ commutes; diagram $\heartsuit$ commutes, for all $A$ and $B$. 
\begin{equation*}
\vcenter{\xymatrix{
T1 \ar[r]^{!}
    \ar@{=}@/_1.0pc/[rr]
& 
1 \ar[r]^{\eta_1}
&
T1
}}\qquad\spadesuit\qquad\qquad\qquad
    \vcenter{\xymatrix@R-1.3pc{
TA \times TB \ar[r]^{\te}
    \ar@{=}[rd]
& 
T(A \times B) \ar[d]^{\chi}
\\
& TA \times TB
}}\qquad\heartsuit
\end{equation*}
\end{definition}
See Figure~\ref{monadtable} for (non-)examples of affine monads.
Regarding preservation of equations by affine monads,
we recall the following result.
\begin{theorem}[\cite{layers}]\label{affinepresdrop}
    Any affine monoidal monad~$T$ on~$\Set$
    preserves strict-drop equations. 
\end{theorem}
Theorem~\ref{affinepresdrop} does not extend to one-drop equations, as follows
from a later result: Example~\ref{zxx-zyx} gives
a one-drop equation whose preservation implies relevance.

\subparagraph*{From drop preservation to affineness}
We proceed to prove the converse of Theorem~\ref{affinepresdrop}: if a monoidal monad $T$ 
preserves a strict-drop equation, then it is affine (Theorem~\ref{thm:drop-affine}).  %
 \begin{lemma}
Let $t_1=t_2$ be a one-drop equation, and $T \colon \Set \rightarrow \Set$ a monoidal monad.
If $t_1=t_2$ holds on $T1$, then $T$ is affine.
\label{lemma:affine1}
\end{lemma}
Since any equation $t_1=t_2$ trivially holds on $1$, by Lemma~\ref{lemma:affine1} we obtain:

\begin{theorem}\label{thm:drop-affine}
Let $\Sigma$ be a signature, and let $t_1=t_2$ be a one-drop equation with $t_1, t_2 \in \Sigma^* V$ and $\Var(t_1) \cup \Var(t_2) = V$. 
Let $T \colon \Set \rightarrow \Set$ be a monoidal monad.
If $T$ preserves $t_1=t_2$, then $T$ is affine.
\end{theorem}
The above theorem is a stronger result than the converse of Theorem~\ref{affinepresdrop}: preservation of \emph{one-drop} equations suffices
for affineness. Equivalently, if a monad is 
\emph{not} affine, then we know that it does not preserve any drop equations in general. 
This generalises 
Gautam's result that $\po$ does not preserve any one-drop equation~\cite{gautam1957validity}.
Other examples of monads that, by Theorem~\ref{thm:drop-affine}, do not preserve one-drop equations 
are  
$M \times X$ for $M$ a non-trivial monoid, and $\ms$ for $\semi$ a non-trivial semiring (see Figure~\ref{monadtable}). 


\begin{remark}
	Theorem~\ref{thm:drop-affine} treats preservation of single equations. 
	Another consequence of Lemma~\ref{lemma:affine1} is that,
	if a monoidal monad $T \colon \Set \rightarrow \Set$ preserves
	a non-empty \emph{set} of equations $E$ that includes a one-drop equation, then it is affine. 
	To see this, note that any equation holds on the algebra $1$; therefore also
	on $T1$, hence $T$ is affine by Lemma~\ref{lemma:affine1}.
\end{remark}
%
%


\subparagraph*{Decidability} 

The previous section establishes the equivalence between affineness of a monad $T$ 
and preservation of one-drop equations.
We now use this result to analyse a more algorithmic question: is it decidable whether
a monad $T$ (presented by finitely many operations and equations) preserves a given equation $t_1=t_2$? 
Unfortunately the answer is negative, which we prove by showing that the question
whether a given monad is affine is undecidable. 

\begin{theorem}\label{thm:affine-undecidable}
The following problem is undecidable: given a finite signature $\Sigma$ and a finite set $E$ of equations, 
is the monad $T$ presented by $(\Sigma,E)$ affine?
\end{theorem}
\begin{proof}
We use an encoding of the following decision problem, which is known to be undecidable~\cite{book1993string}: 
given a finite presentation $(G,R)$ of a monoid $\mathcal{M}$,
is $\mathcal{M}$ trivial?

Let  $(G,R)$ be a finite presentation  of a monoid $\mathcal{M}$. 
Let $\Sigma$ be the set of unary operations $f_g$, for $g \in G$, 
and $E$ the set of equations
$f_{g_1}(f_{g_2}( \dots f_{g_n}(x) \dots))= f_{h_1}(f_{h_2}( \dots f_{h_k}(x) \dots))$, for each $(g_1 \dots g_n, h_1 \dots h_k) \in R$. Note that $(\Sigma,E)$ corresponds to the theory of $\mathcal{M}$-actions; let $T$ be the monad presented by this theory.
Now, one can show that $T1$ is isomorphic to $\mathcal{M}$, 
and thus that $T1=1$ iff $\mathcal{M}$ is trivial. 
\end{proof}
%
%

Using the equivalence between preserving a class of equations and affineness, together with the latter being undecidable, we obtain a general result on  equation preservation.

\begin{corollary}
The following problem is undecidable: given a finite theory $(\Sigma,E)$ and an equation $t_1=t_2$, does the monad $T$ presented by $(\Sigma,E)$ preserve $t_1=t_2$?
\end{corollary}

\section{Relevant Monads and Dup Equations}\label{sec:dup}

We now relate so-called \emph{dup} equations, featuring duplications of variables,
to \emph{relevant} monads. 

\begin{definition}[Dup equations]\label{def:dup}
An equation $t_1=t_2$ is called
\begin{enumerate*}
	\item \emph{dup} when at least one variable appears more than once in $t_1$ or in $t_2$;
    \item \emph{2-dup} when it is dup and each variable appears at most twice in~$t_1$ or~$t_2$;
	\item \emph{strict-dup} when it is not linear and each variable of $\Var(t_1) \cup \Var(t_2)$ 
	appears at least once in $t_1$ and in $t_2$.
\end{enumerate*}
\end{definition}
Equivalently, an equation is strict-dup when it is not drop, and some variable appears at least twice
in $t_1$ or $t_2$. 
Every strict-dup equation is a dup equation. See Figure~\ref{venn} for an overview of dup and drop equations.
%
For example, the law of idempotence $x = x \cdot  x $ is a strict-dup equation. 
Distributivity of $\cdot $ over $+$, written $x\cdot (y+z)=x\cdot y+x\cdot z$ is strict-dup as well.
The equation $x\cdot (y\cdot y) = y\cdot y$ is dup, because $y$ is duplicated, but it is not strict-dup. 

\noindent
Relevant monads are introduced by Kock in~\cite{kock1972bilinearity}, and extensively studied by Jacobs in~\cite{jacobs1994semantics}.
\begin{definition}\label{definition:relevant}
    A monoidal monad $T \colon \cat \rightarrow \cat$ on a Cartesian monoidal category $\cat$ 
    is \emph{relevant} if one of the following two equivalent conditions hold: diagram $\spadesuit$ commutes for all objects $A$;  diagram $\heartsuit$  commutes for all objects $A,B$:
\begin{equation*}\label{phi_o_chi_def}
    \vcenter{\xymatrix@C+2pc@R-1.3pc{
TA \ar[r]^-{\Delta}
	\ar[rd]_-{T \Delta}
& 
TA \times TA \ar[d]^-{\te}
\\
& T(A \times A)
}}\qquad\spadesuit\qquad\qquad
    \vcenter{\xymatrix@C+2pc@R-1.3pc{
T(A \times B) \ar[r]^-{\chi}
	\ar[rd]_{\id}
& 
TA \times TB \ar[d]^-{\te}
\\
& T(A \times B)
}}\qquad\heartsuit
\end{equation*}
\end{definition}

\noindent
See Figure~\ref{monadtable} for  examples of relevant monads.
We recall the following result.
\begin{theorem}[\cite{layers}]
    Any relevant monad on $\Set$ preserves strict-dup equations.
\end{theorem}
Does this theorem extend to an equivalence? Can a non-relevant monad preserve a dup equation? Again, Gautam provides an answer in the case of the powerset, which is not a relevant monad: $\po$ does not preserve dup equations. 
We now study the general question for a monoidal monad $T$. For this purpose, we define a framework of string-like diagrams to easily represent the application of $T$ on categorical objects.
\subparagraph*{Diagrammatic framework}

Our diagrams are inspired from well-known graphical representations for monoidal categories. The concept of \emph{functorial boxes} is introduced by Cockett and Seely in \cite{cockett1999linearly} then further examined by Melli\`es in \cite{mellies2006functorial}, where functors are represented as boxes surrounding morphisms and objects. In a similar fashion as our framework, monoidal properties allow to gather several objects inside a single sleeve. More details on this representation are given by McCurdy~\cite{mccurdy2011graphical}, although he chooses to focus on monoidal functors satisfying the Frobenius property, which we do not assume here. Let us first summarise a few central ideas of our diagram calculus.

\begin{wrapfigure}{r}{4.6cm}
\vspace{-.1cm}
\begin{tabular}{@{}P{2cm} P{2cm}@{}}
$TX \times TY$ & $T(X \times Y)$
\\
\includegraphics[scale=0.375]{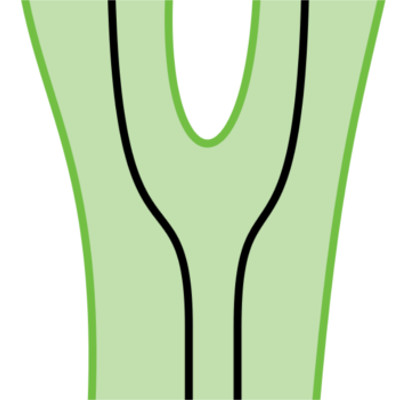} 
&
\includegraphics[scale=0.375]{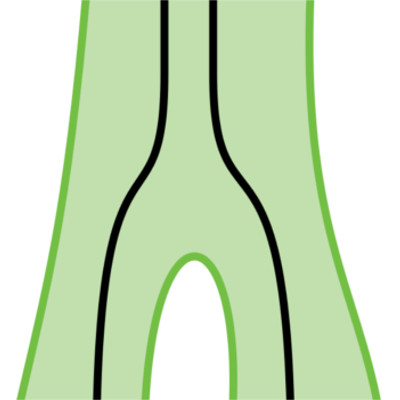} 
\\
$T(X \times Y)$ & $TX \times TY$
\\
\hline $\chi$ &  $\te$
\end{tabular}
\vspace{-.4cm}
\end{wrapfigure}
An object $X$ of our category is now represented by a thread (or `wire'), and the application of $T$ on this object results in a `sleeve' covering it. We read a diagram from bottom to top and represent products implicitly as horizontal adjacency. For instance, the morphism $\chi\colon  T(X \times Y) \to TX \times TY$ is modelled by a cup-like shape where one sleeve containing two objects splits into two sleeved objects. $\te$ is modelled in the opposite way and merges two sleeved objects into a single sleeve.

\begin{wrapfigure}{r}{4.6cm}
\vspace{-.5cm}
\begin{equation*}
\includegraphics[align=c,scale=0.11]{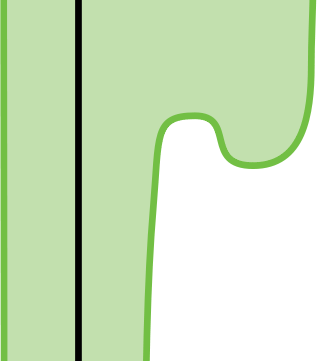} 
 \ =\ 
\includegraphics[align=c,scale=0.11]{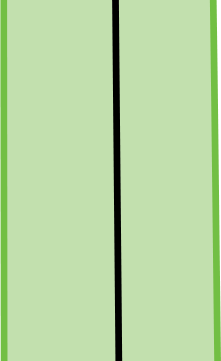} 
\end{equation*}
\vspace{-.5cm}
\end{wrapfigure}
Note that the object $1$ is not represented in our diagrams. By the isomorphism $X\times 1 \simeq 1$, we can imagine the presence of $1$ as a vertical thread anywhere on the diagram without affecting calculations. Some deformations of the outline of sleeves are allowed: for instance, the equality in the diagram on the right corresponds to the right unitality of a monoidal functor. 
Note that the neither the product nor the unitor $\rho$ is explicitly represented in the diagram. 
\begin{wrapfigure}{r}{6.5cm}
\vspace{-.3cm}
\begin{tabular}{P{2cm} M{0.5cm} P{2cm}}
\includegraphics[align=c,scale=0.115]{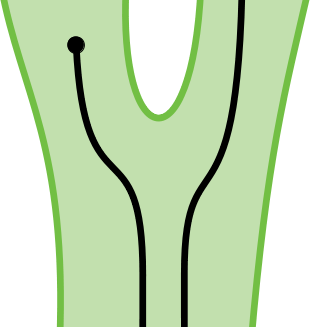} 
&$=$&
\includegraphics[align=c,scale=0.115]{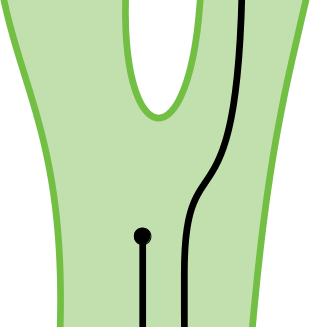} 
\\
\hline
 $(T!\times \id) \circ \chi$ & $=$ &  $\chi \circ T(!\times \id) $
\end{tabular}
\vspace{-.5cm}
\end{wrapfigure}

We can `delete' an object by mapping it to the final object $1$, represented as an unfinished vertical thread. Naturality of $\chi$ and $\te$ allow to `pull' these threads to the bottom of the diagram, as shown in the equation on the right.

The following result means that unfinished threads may be ignored.
\begin{lemma}\label{lemma:fgsleeve}
If we have any of the two equalities $(i)$ or $(ii)$ below, then $f=g$.
\begin{align*}
(i)
    \begin{tikzpicture}[baseline=(current bounding box.center)]
        \draw (0, 0) node[inner sep=0] {
            \includegraphics[scale=0.07]{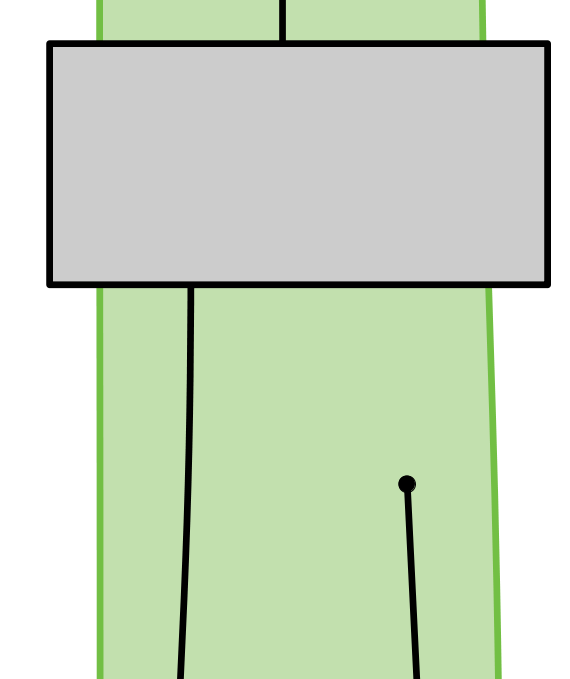}};
        \draw (0,0.45) node{f};
    \end{tikzpicture}
    \ = \ 
    \begin{tikzpicture}[baseline=(current bounding box.center)]
        \draw (0, 0) node[inner sep=0] {
            \includegraphics[scale=0.07]{Sleeve_pics/drop2.png}};
        \draw (0,0.45) node{g};
    \end{tikzpicture}
    & \qquad(ii)
        \begin{tikzpicture}[baseline=(current bounding box.center)]
        \draw (0, 0) node[inner sep=0] {
            \includegraphics[scale=0.07]{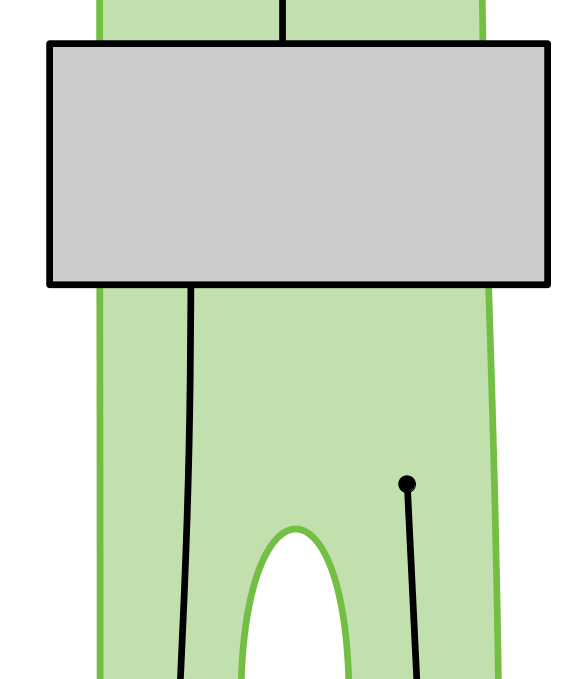}};
        \draw (0,0.45) node{f};
    \end{tikzpicture}
    \ = \ 
    \begin{tikzpicture}[baseline=(current bounding box.center)]
        \draw (0, 0) node[inner sep=0] {
            \includegraphics[scale=0.07]{Sleeve_pics/drop1.png}};
        \draw (0,0.45) node{g};
    \end{tikzpicture}
\end{align*}
\end{lemma}


\begin{wrapfigure}{r}{6.4cm}
\vspace{-0.85cm}
\begin{equation}\label{relevancediag}
\begin{tabular}{P{2cm} M{0.5cm} P{2cm}}
$T(X \times Y)$ && $T(X \times Y)$
\\
\includegraphics[align=c,scale=.5]{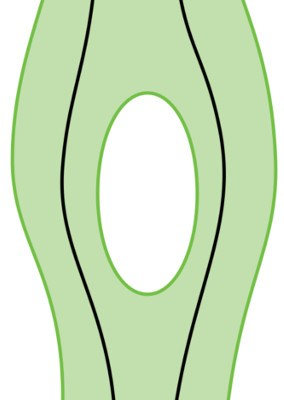} 
&\ $=$\ &
\includegraphics[align=c,scale=.5]{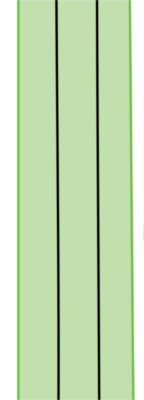} 
\\
$T(X \times Y)$ && $T(X \times Y)$
\\
\hline
 $\te \circ \chi$ & $=$&  $\id$
\end{tabular}
\end{equation}
\vspace{-1.1cm}
\end{wrapfigure}
Finally, we focus on the composition $\te \circ \chi$, which splits a sleeved product, then reunites both components into one sleeve. Recall that relevance means 
that this composition yields the identity (Definition~\ref{definition:relevant}).
In \cite{mccurdy2011graphical}, this diagrammatic equality is presented as a property of connectivity of functorial regions. We like to describe the graphical aspect of this property as follows: applying $\chi$ followed by $\psi$ results in a `bubble' in our sleeve, surrounded by two threads representing arbitrary objects, and relevance allows to pop this bubble, as in~\eqref{relevancediag}. In the rest of this section, we develop a method to reduce complex equational problems to this `bubble' property.

\subparagraph*{Proving relevance from dup preservation}
First, we consider the simplest dup equation: idempotence
of a binary operation. Assuming that $T$ preserves it, our strategy
is to define an algebra $\algb$ and an operation $m$ such that
$m(x,x)=x$, and then to derive the relevance of $T$ from the
preservation of the idempotence of $m$. For such an algebra (whose
carrier is written $A$), we draw the following diagrams after this paragraph.
The grey
box represents our binary idempotent operation $m$, hence the leftmost
diagram represents the term $m(x,x)$ on the lifted algebra $\widehat{T}
\algb$. By preservation of idempotence, it must be equal to the
identity modelled by the right diagram in the leftmost equality.
In order to derive the property of relevance from this equality,
we conveniently choose $A$ and $m$. For any two sets $X,Y$, we
define $A \equiv X \times Y$ and
$m((a,b),(c,d)) \equiv (a,d)$.
Categorically~$m = (\pi_1 \times
\pi_2)$. It is idempotent: $m((a,b),(a,b))= (a,b)$.
Hence for this algebra,
the left equality becomes the right one:
\begin{center}
\begin{tabular}{P{2cm} M{0.1cm} P{2cm}}
$TA$ && $TA$
\\
\includegraphics[align=c,scale=0.6]{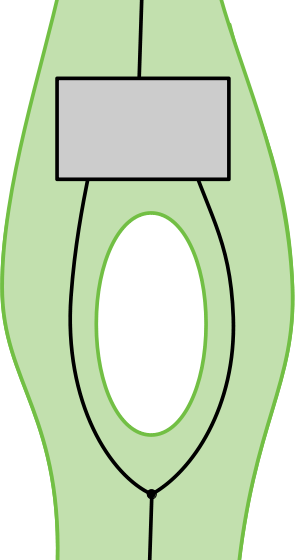} 
&$=$&
\includegraphics[align=c,scale=0.890]{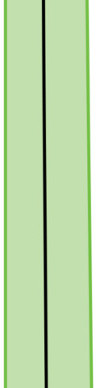} 
\\
$TA$ && $TA$
\\
\hline
 $T(m) \circ \te \circ \Delta$ & $=$&  $\id$
    \end{tabular}
    \qquad
\begin{tabular}{P{2cm} M{0.1cm} P{2cm}}
$T(X\times Y)$ && $T(X\times Y)$
\\
\includegraphics[align=c,scale=0.9]{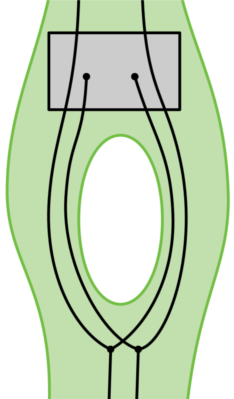} 
&$=$&
\includegraphics[align=c,scale=0.9]{Sleeve_pics/id_morphism2.png} 
\\
$T(X\times Y)$ && $T(X\times Y)$
\\
\hline
 $T(m) \circ \te \circ \Delta$ & $=$&  $\id$
\end{tabular}\end{center}
Pulling down the threads corresponding to deleted objects,
    we obtain diagram \eqref{relevancediag}.

\begin{theorem}\label{relevantiffdup}
Let $T$ be a monoidal monad. If $T$ preserves $m(x,x)=x$, then $T$ is relevant.
\end{theorem}
\begin{proof}
We show here the categorical version of our diagrammatic proof:
\begin{align*}
\id & \ =\ T(m) \circ \te \circ \Delta &x*x \text{ holds on }\widehat{T}\algb\\
 &\ =\  T(m) \circ \te \circ \chi \circ T\Delta &\text{see below}\\
 &\ =\  T(\pi_1 \times \pi_2) \circ \te \circ \chi \circ T\Delta &\text{definition of }m\\
 &\ =\  \te  \circ (T\pi_1 \times T\pi_2) \circ \chi \circ T\Delta &\text{naturality}\\
 &\ =\  \te  \circ \chi \circ T(\pi_1 \times \pi_2)  \circ T\Delta &\text{naturality}\\
    &\ =\  \te \circ \chi & & 
\end{align*}
In the second step, we used that 
$\Delta_{TA} = \chi_{A,A} \circ T \Delta_A$, a property that holds trivially for any 
monoidal monad on $\Set$ and set $A$ (see for instance \cite{jacobs1994semantics}).
\end{proof}

\vspace{-.2cm}
\begin{wrapfigure}{r}{0.3\textwidth}
\vspace{-1cm}
\begin{gather}\label{yxz}
    (yx)z = (y(xx))z. \\
    \vcenter{\hbox{\includegraphics[scale=0.08]{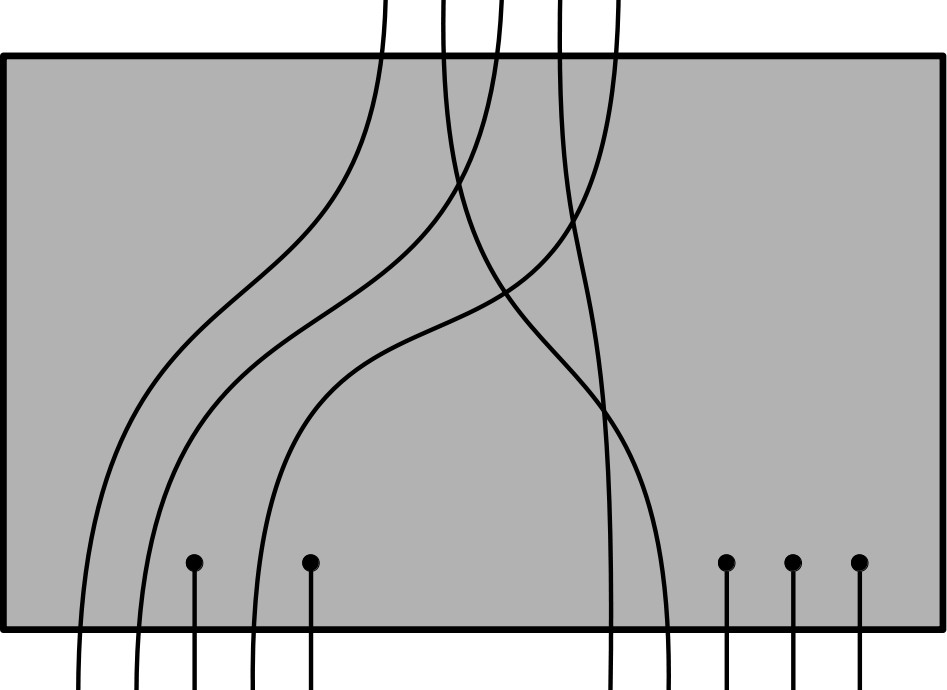}}}
    \label{yxzbox}\\
    \begin{tikzpicture}[baseline=(current bounding box.center),
            level distance=0.3cm, sibling distance=0.9cm]
        \coordinate
        child {
            child { node{$y$} }
            child {
                child { node{$x_1$} }
                child { node{$x_2$} }
            }
        }
        child { node{$z$} };
    \end{tikzpicture}
    \label{yxztree}
\end{gather}
    \vspace{-.7cm}
\end{wrapfigure}%
Our method to show relevance from idempotence preservation
    can be applied to a more general class of equations. Let us now consider a term $t$ that only contains binary operations and no variable duplication. We focus on equalities of the form~$t[x] = t[x \cdot x]$ with $\cdot \in \Sigma$, in other words where both sides only differ by one variable duplication. We sketch the method by treating a concrete example:
    equation~\eqref{yxz} on the right.
Note that this equality is of the desired form with~$t[N] \equiv (yN)z$.
Assume $T$ preserves~\eqref{yxz} and let~$X$ be a set.
Once again, we define a convenient algebra: its carrier is~$X^5$, and every  binary operation of $\Sigma$ is interpreted with the morphism~$m\colon  X^5 \times X^5 \to X^5$, graphically represented in~\eqref{yxzbox}.
The map $m$ can be categorically defined as $\langle \pi_1, \pi_7, \pi_2, \pi_6, \pi_4 \rangle$, but we will rather describe it as making a series of connections (seen as wires) between the components of its output and of its left and right inputs. The outputting wires are respectively labelled~$L^*$, $R^*$, $LR^*$,
    $RL^*$ and $LRL^*$.
We represent the syntax tree of~$t$ in~\eqref{yxztree} ($x_1$ and $x_2$ representing the two duplicates of $x$). Encoding locations in the tree as words with letters $L$ and $R$, 
	 the node where the duplication occurs is labelled $LR$. Let us now picture this tree with $m$-boxes on each node. By construction, the label we gave to each outputting wire describes the set of locations in the tree that are traversed by this thread.

\begin{lemma}\label{msatisfies}
The binary operation $m$ on $X^5$ satisfies the equation \eqref{yxz}.
\end{lemma}
\begin{proof}
As~$(yx)z = (y(xx))z$ is preserved by $T$, the equality on the left in the diagram below
 holds.  From this, pulling down the unfinished threads,
        we get the equality on the right.
\begin{multicols}{2}
\noindent
\begin{align*}\label{bigsleeve}
    \includegraphics[align=c,scale=0.035]{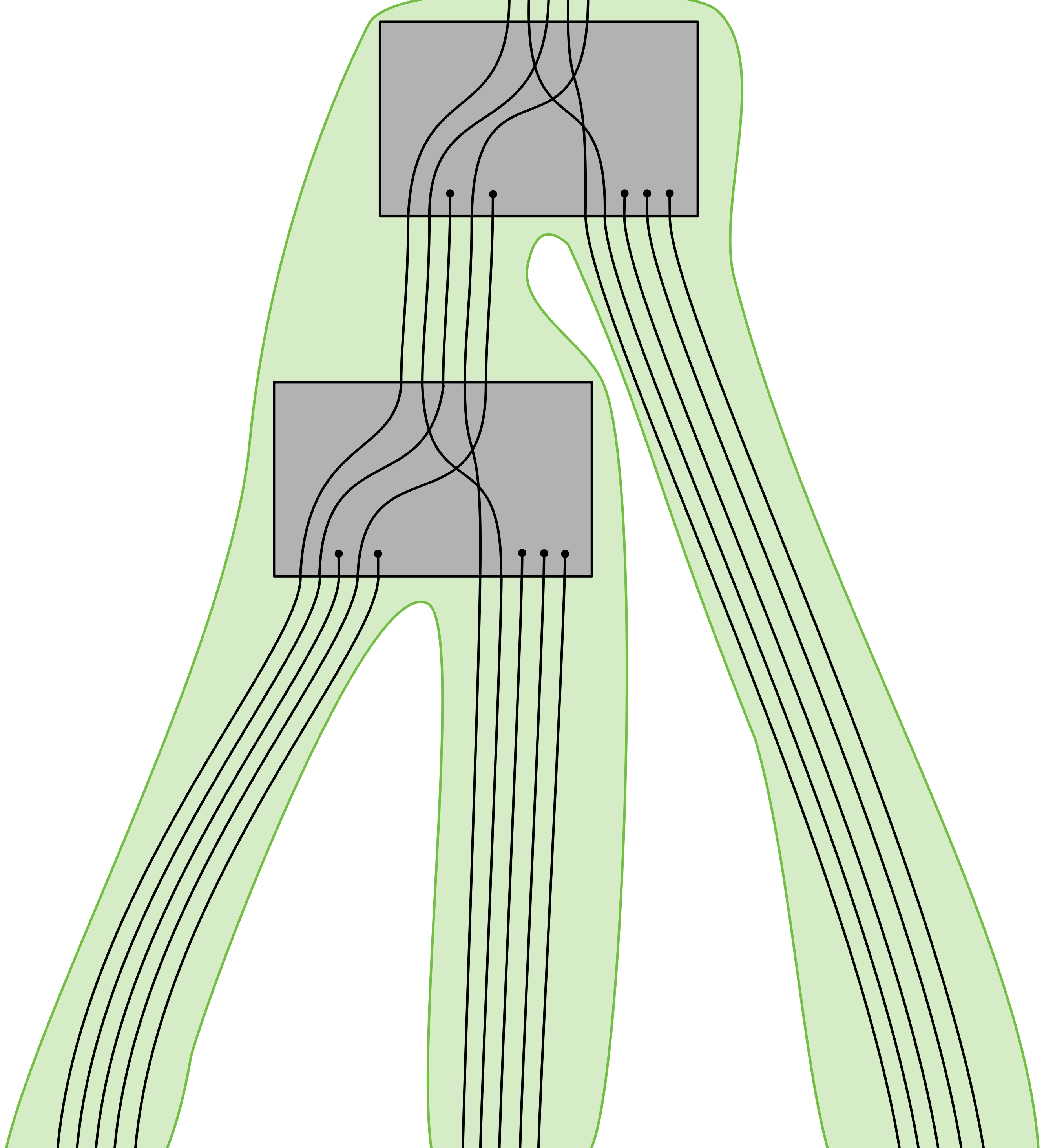} & =  
    \includegraphics[align=c,scale=0.035]{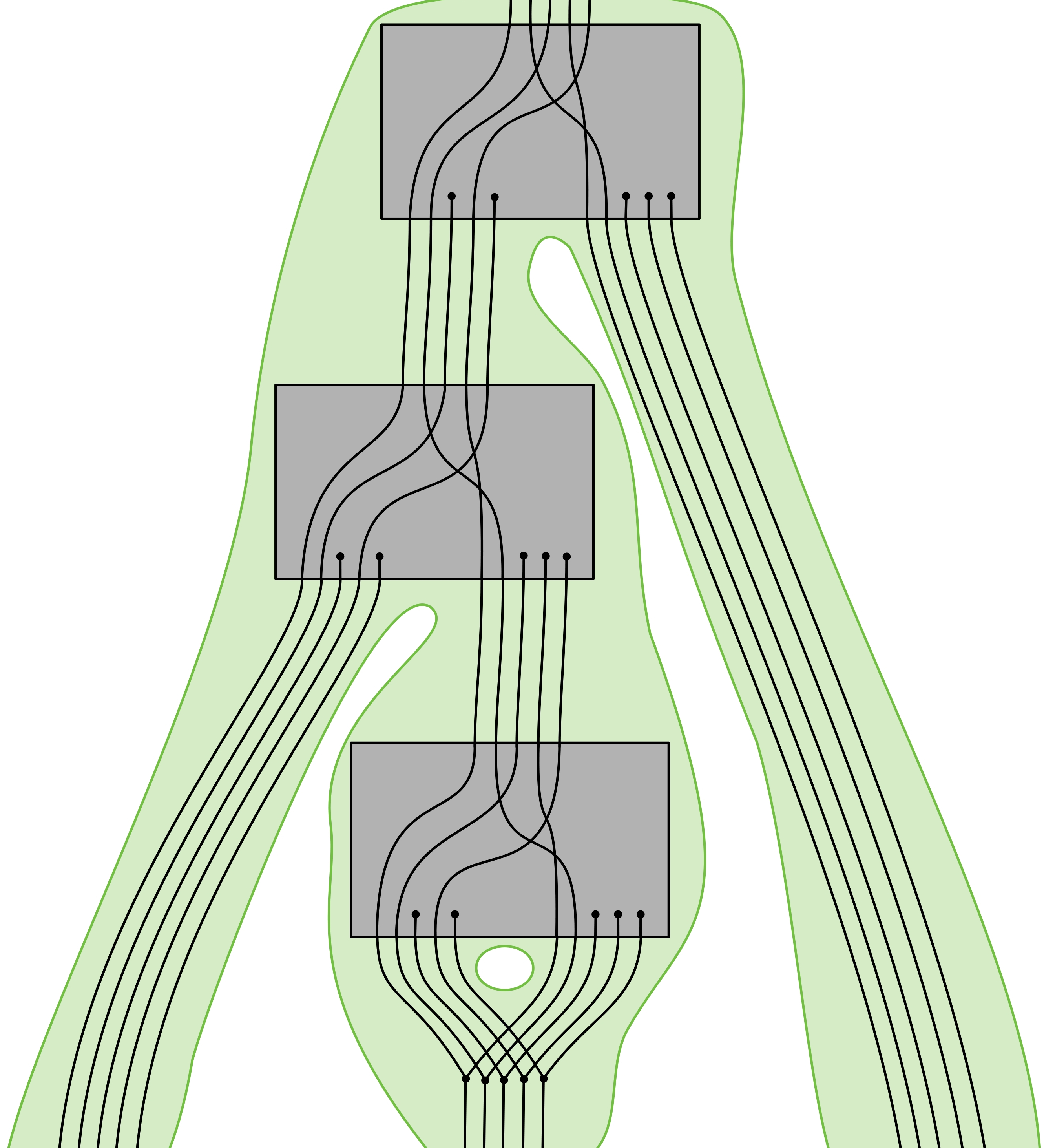} 
    \\
    y \qquad \quad x \quad \quad z \quad & \qquad   y \quad \quad x \qquad \quad z  \nonumber
\end{align*}
\columnbreak
    \begin{equation*}\label{bigsleeve2}
    \includegraphics[align=c,scale=0.035]{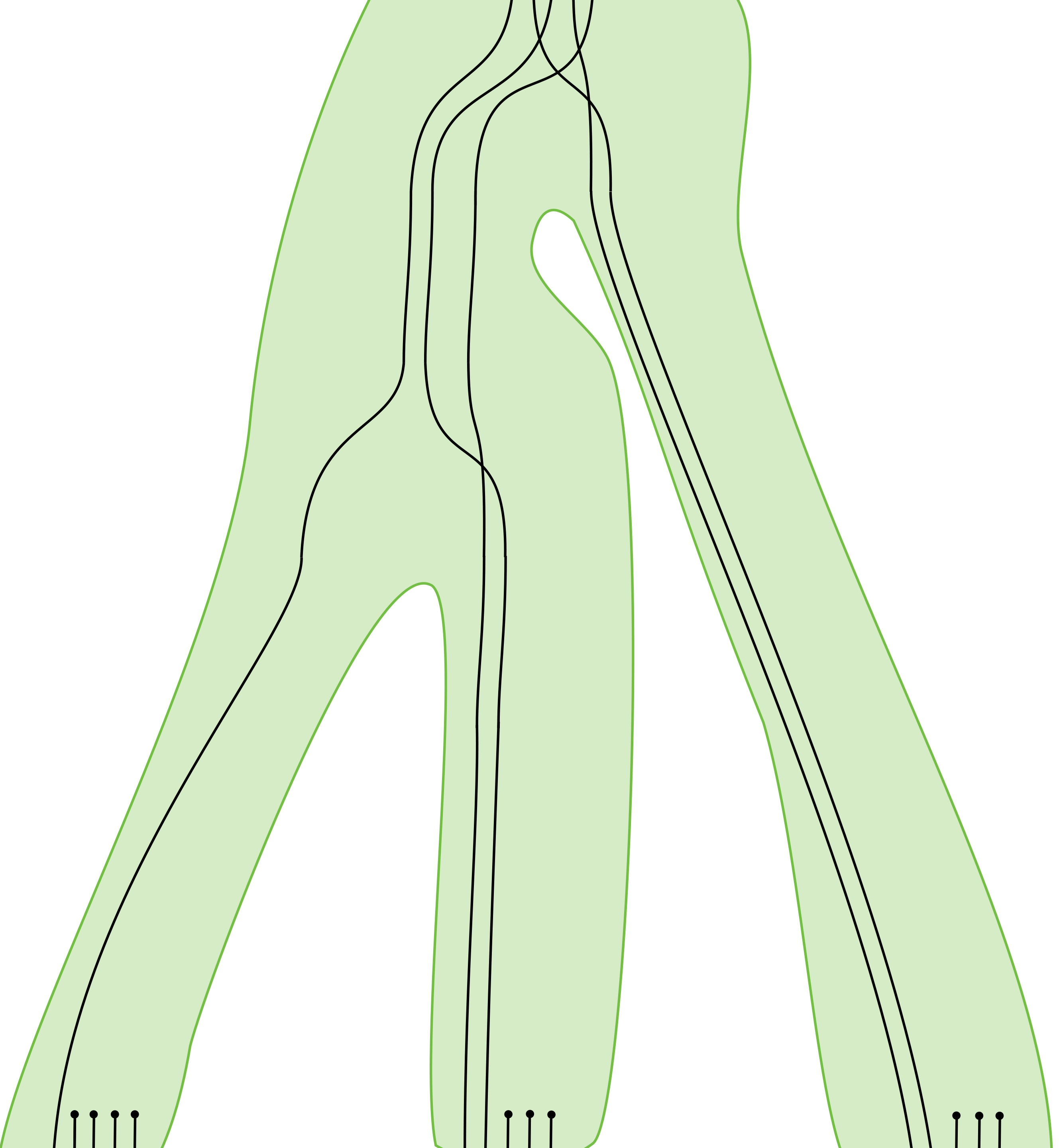}  =  
    \includegraphics[align=c,scale=0.035]{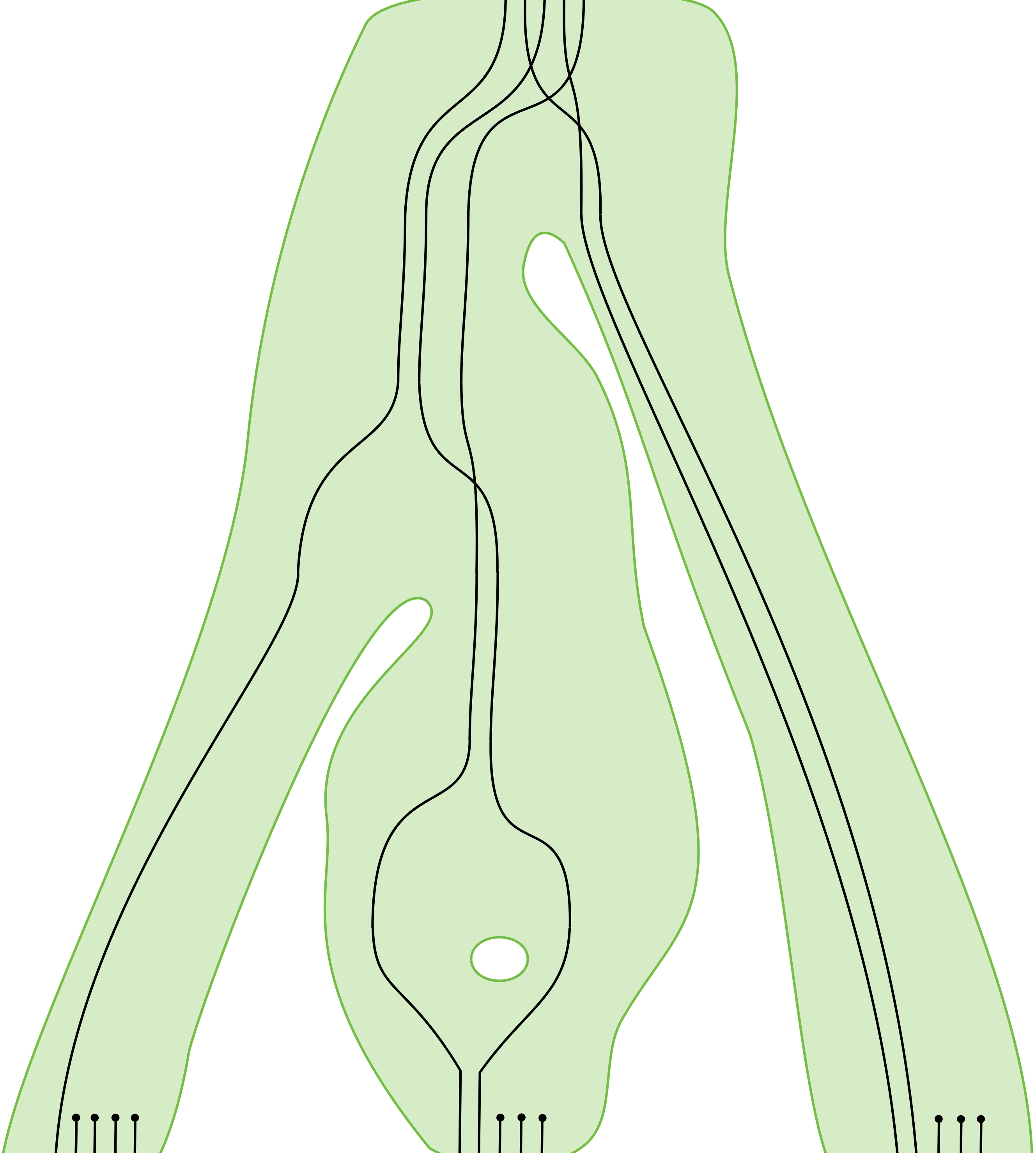}
\end{equation*}%
\end{multicols}%
\vspace{-1cm}
As in the case of idempotence case, we now have two wires `wrapping' the bubble ($LR^*$ and~$LRL^*$). By Lemma~\ref{lemma:fgsleeve}, we can ignore the unfinished wires and obtain diagram~\eqref{relevancediag}. 
\end{proof}
For any equation of the form $t[x]=t[x \cdot x]$, we can design $m$ satisfying the equality and allowing to show relevance. We only have to make sure that the box representing $m$ features two wires realising the connections $L^*$ and $R^*$, as well as two others labelled $wR^*$ and $wL^*$, where $w$ is the word on $\{L,R\}^*$ describing the location of the variable duplication in $t$. 
    
    \begin{theorem}\label{thm:sleevebinary}
        $T$ is relevant if it preserves~$t[x]=t[x\cdot x]$
            when~$t$ only contains binary ops.
    \end{theorem}

Although this result may seem too specialised, the process described above actually applies to other equations. Recall that our strategy relies on defining a convenient algebra to derive relevance from the preservation of a particular equation. If $T$ preserves $(x+x)\cdot y=x \cdot y$, in particular $T$ preserves it on an algebra where $\cdot $ and $+$ are interpreted as identical, hence why we only needed to define one binary operation in the previous paragraph. We may generalise this even further to treat $n$-ary operations: if $f(x,x,z)=x \cdot z$ is preserved by $T$, then in particular it is on an algebra where $f(x,x,z)=(x\cdot x)\cdot z$ (as long as we can define such an algebra where the considered equation also holds). Since we have treated $(x\cdot x)\cdot z=x\cdot z$ above, our conclusion also applies to  $f(x,x,z)=(x\cdot x)\cdot z$. We have obtained the property of relevance from any possible case featuring binary operations, thus we also obtain for free the case of $n$-ary operations (where $n>1$).

    \begin{theorem}\label{thm:sleevenary}
Let $t$ be a term without constants. $T$ is relevant if it preserves $t[x]=t[x \cdot x]$.
    \end{theorem}

We can generalise this even further to cover equations outside the strict-dup class. If we use the above approach to treat an equation $t[x]=t[x \cdot x]$, it turns out we can slightly modify the equation without affecting our result. Consider the example $(yx)z=(y(xx))z$ again. At the end of our process, the only component of the $y$ that is connected to the output is the first one (through the $L^*$ wire). By construction, adding another iteration of $m$ with $y$ on its left input would not change this fact. Let us then substitute $y$ with $yv$ in the equation (for $v$ any new variable). The position of $v$ is coded as the word $LLR$, which does not belong to the language of any of our outputting wires, therefore $v$ has no influence on the matching of the outputs. In other words, the definition of $m$ allowing to prove relevance from the preservation of $(yx)z=(y(xx))z$ also applies to $((yv)x)z=(y(xx))z$, which is a one-drop equation. One could even substitute $v$ with a more complicated term to obtain another equation, whose preservation would still lead to relevance. This last theorem applies therefore to many equalities outside the case $t[x]=t[x \cdot x]$, even though the exact class described by these modifications is cumbersome to define.

\begin{example}\label{zxx-zyx}
Because Thm.~\ref{thm:sleevenary} applies to $z(xx)=zx$, it is also the case for $z(xx)=(zy)x$.
\end{example}

%
%
\subparagraph*{n-relevance}

We have shown that the preservation of~$x \cdot x = x$ implies relevance.
What about the preservation of~$f(x,x,x) = x$?
We will see that it does not imply relevance,
    but rather a weaker property which we will call~$3$-relevance.
To define~$n$-relevance in general,
    we introduce $n$-ary variations of existing maps:
        $\Delta^n$ is the $n$-times duplication operator $X \to X^n$
        and~$\chi^n \equiv \langle T \pi_1, \dots, T \pi_n \rangle$.
Now, we say $T$ is $n$-relevant iff~$\psi^n \circ \Delta^n = T\Delta^n$.
Or equivalently iff~$\te^n \circ \chi^n = \id$.
\begin{proposition}\label{prop:relevance-weakening}
Relevance implies $n$-relevance for~$n \geq 2$.
\end{proposition}
For a commutative monoid~$M$ and the monad $M\times X$
 from Example~\ref{ex:mtimesmonad},
    we have
    $   (\te^n \circ \chi^n)(v, (x_1, \ldots, x_n))
        = \te^n((v, x_1), \ldots, (v, x_n)) 
        = (v^n, (x_1, \ldots, x_n))$
    and so~$M\times X$
        is~$n$-relevant
        iff~$w^n = w$ for all~$w \in M$.
    Hence monads may be~$n$-relevant
        but not~$m$-relevant for any~$n>m$.
%
For affine monads the difference disappears:
\begin{proposition}\label{prop:relevant-and-affine}
Given any $n \in \mathbb{N}$,  if $T$ is $n$-relevant and affine, then $T$ is relevant.
\end{proposition}


As promised, we relate~$n$-relevance to preservation of equations:
\begin{theorem}\label{thm:nrelevance}
Assume $\Sigma$ features an $n$-ary operation $f^n$. 
$T$ preserves $f^n(x, \dots, x)=x$ if and only if $T$ is $n$-relevant.
\end{theorem}
We have seen that the preservation of some~2-dup non-drop
    equations, like~$xx=x$ and~$(xx)y = xy$,
    implies relevance. However, the following statement shows
    that there are 2-dup non-drop equations
    whose preservation does not imply relevance.
    Indeed, $\mzt$ is not relevant, but:

\begin{proposition}\label{prop:multiset-preserves}
	The generalised multiset monad~$\mzt$ 
	preserves $x ( y y) = yx$. 
\end{proposition}

\subparagraph*{Preservation of discerning equations}
\label{discerning}
We present a class of equations for which relevance
    is necessary for preservation.
%
    A \emph{2-discerning equation}~$t_1=t_2$
    is a 2-dup non-drop equation,
    where only one variable, say~$x_1$ out of~$x_1, \ldots, x_n$
    is duplicated and only one side, say~$t_2$,
    which can distinguish the places where~$x_1$ is duplicated
        in the following sense:
        the linear equation~$s_2=s_2'$ in~$x_1, x_1', x_2, \ldots, x_n$
        fixed by~$t_2 = s_2[x_1'/x_1]$
            and~$s_2' = s_2[x_1/x_1', x_1'/x_1]$
            is not derivable from~$t_1=t_2$.
\begin{example}
The equation~$x(yy) = yx$ is \emph{not}~2-discerning
    as it implies~$xy=yx$ and 
    in particular~$x(yy') = x(y'y)$.
On the other hand~$y(xy)= yx$
    is 2-discerning.
This requires one to show that~$y(xy')=y'(xy)$
    isn't derivable from~$y(xy)=yx$,
    which is easily seen
    by noting that all terms equal to~$y'(xy)$
    must start with~$y'$ as well.
In fact, all of the following equations are~2-discerning,
    which are essentially all the remaining candidates
    on two variables:
	$(yy) x = yx,
    (yx) y = yx,
    (xy) y = yx,
    y(y x) = yx,$ and 
    $y(x y) = yx$.
\end{example}
Theorem~\ref{thmdiscerning} states that relevance is equivalent to preservation of 2-discerning equations.
In the proof~\cite{parlant2020preservation}, we assume that $T$ is finitary---that is, $T$ is presentable as an algebraic theory~$\Th$,
    i.e.~$TX$ is the free $\Th$-algebra over~$X$;
    $Tf$ for~$f\colon X \to Y$ maps a term~$M \in TX$
        to~$M[x/f(x)]$;
    $\eta_X$ maps~$x$ to the term~$x$
        and~$\mu_X$ maps a term over terms to the collapsed term.
        The argument for arbitrary monads (which are presented
        by infinitary algebraic theories)
        is similar, but heavier on paper.
        We first relate relevance of a finitary monad
        to its presentation; this algebraic characterisation can be found with slightly different notation in Figure 7 of \cite{kammar2012algebraic}.

\begin{proposition}\label{propalgrel}
Suppose~$T$ is a monoidal monad on~$\Set$
    presented by an algebraic theory~$\Th$.
Then~$T$ is relevant iff 
    for every~$n$-ary operator~$f$ of~$\Th$
        we have
   $
        \mtr{f} ((x_{ij})_{ij}) \ =\  \vec{f} ( (x_{ii})_i) )
    $,
    where~$(x_{ij})_{ij}$ is a~$n\times n$ matrix
        over any set~$X$,
        $\vec{f}( (y_i)_i ) \ \equiv\  f(y_1, \ldots, y_n)$
        and $\mtr{f} ( (y_{ij})_{ij} )
            \ \equiv \ f(
                f(y_{11}, \ldots, y_{1n}), \ldots,
                f(y_{n1}, \ldots, y_{nn}))$.
\end{proposition}
For instance, in an algebraic theory presenting a relevant monoidal monad,
    we must have~$f^2=f$ for any unary~$f$
        and~$g(g(a,b),g(c,d)) = g(a,d)$
            for any binary~$g$. 
\begin{theorem}\label{thmdiscerning}
Suppose~$t_1=t_2$ is a~2-discerning equation
    and~$T$ is a monoidal monad on~$\Set$.
Then~$T$ is relevant if and only if~$T$ preserves~$t_1=t_2$.
\end{theorem}

\section{Related work}
The notions of relevant and affine monads are systematically studied in~\cite{kock1972bilinearity,jacobs1994semantics},
but those works do not treat preservation of equations.
Pioneering work on the preservation of algebraic features by a monad goes back to~\cite{gautam1957validity}. Without any notion of category theory, Gautam pinpoints exactly which equations are preserved by the powerset monad and highlights the importance of variable duplications or deletions.

Methods for combining the features of two monads have been discussed in several papers. Hyland et al.\ work out two canonical constructions in \cite{hyland2004combining}, the \emph{sum} and the \emph{tensor} of monads. Our work is closer to King and Wadler's study in \cite{king1993combining}, as they use distributive laws. Later, Manes and Mulry show in \cite{manes2007monad} and \cite{manes2008monad} that a correspondence between categorical properties of one monad and algebraic features of another may lead to a distributive law. Their work also sheds light on the importance of a monoidal structure in the problem of lifting signatures. This approach is then generalised to affine and relevant monads in \cite{layers}. Our contributions extend this work by showing on the one hand that some of the sufficient conditions of \cite{layers} are
necessary (Theorems~\ref{thm:drop-affine}, \ref{thmdiscerning}, \ref{thm:sleevenary}), and on the other hand that some algebraic features require finer categorical conditions to be preserved (Theorem~\ref{thm:nrelevance}).

In Section~\ref{discerning}, we rely on presenting monads with algebraic theories in order to study their composition. This approach is related to Zwart and Marsden's method in \cite{ZwartM19}. In that paper, the authors elaborate on the concept of composite theory, introduced and studied by Pirog and Staton in \cite{pirog2017backtracking}, and establish the absence of distributive laws under different algebraic conditions on the considered monads. The main difference is that in~\cite{ZwartM19}, the focus is on combining algebraic theories,
whereas the current paper relates algebraic structures to categorical properties of monads (relevance and affineness). These properties, along with other categorical conditions, are characterised algebraically in a similar manner to Proposition~\ref{propalgrel} by Kammar and Plotkin in~\cite{kammar2012algebraic}, and our work connects this  characterisation with the preservation of certain equations.

\section{Conclusions and Future work}

We have systematically related preservation of drop and dup equations to affineness and relevance of monoidal monads,
respectively. There are several avenues for future work. 
First, this paper focuses on monads on $\Set$. Generalising this work to monads on arbitrary Cartesian monoidal categories would be a natural follow-up and would require to use a more abstract notion of equation.
Second, our work focuses specifically on monoidal liftings. Therefore, the distributive laws that we obtain from preservation 
are of a specific shape. It would be interesting to algebraically characterise which distributive laws arise in this way. 
Finally, we would like to go beyond the world of monoidal liftings, possibly allowing a non-affine monad to preserve drop equations for such non-canonical liftings. This could give a new perspective on the construction of distributive laws. 



\bibliography{monadsbib}

\appendix

\section{Details of Section \ref{sec:prelim}}
\label{sec:appendix_prelim}
A few definitions were omitted from our preliminary section, we present here the details.
\begin{definition}[Monoidal Functor]
Let $\cat$ be a Cartesian monoidal category. 
A \emph{(lax) monoidal functor} is an endofunctor $F\colon \cat\to\cat$ together with natural transformations $\psi_{X,Y}\colon  FX\times FY\to F(X\times Y)$ and $\psi^0\colon  1\to F1$ satisfying the diagrams:

\hspace{-0.5cm}
\begin{tabular}{l c}
\begin{minipage}[t]{6cm}
\begin{equation}
\vcenter{
\xymatrix@C=7ex
{
FX\times 1\ar[r]^{\id_{FX}\times \psi^0}\ar[r]^{\id_{FX}\times \psi^0}\ar[d]^{\rho_{FX}} & FX\times F1\ar[d]_{\psi_{X,1}} \\
FX & F(X\times 1)\ar[l]_{F\rho_X}
}}\label{diag:monoidal:left_unit}\tag{MF. 1}
\end{equation}
\end{minipage}
&
\begin{minipage}[t]{6cm}
\begin{equation}\vcenter{
\xymatrix@C=7ex
{
1\times FX\ar[r]^{\psi^0\times \id_{FX}}\ar[r]^{\psi^0\times \id_{FX}}\ar[d]^{\rho'_{FX}} & F1\times FX\ar[d]_{\psi_{1,X}} \\
FX & F(1\times X)\ar[l]_{F\rho'_X}
}}\label{diag:monoidal:right_unit}\tag{MF. 2}
\end{equation}
\end{minipage}
\\
\multicolumn{2}{c}{ \begin{minipage}[t]{12cm}
\begin{equation}
\vcenter{
\xymatrix@C=12ex
{
(FX\times FY)\times FZ\ar[r]^{\alpha_{FX,FY,FZ}} \ar[d]^{\psi_{X,Y}\times \id_{FZ}}& FX\times (FY\times FZ) \ar[d]_{\id_{FX}\times \psi_{Y,Z}}\\
F(X\times Y)\times FZ\ar[d]_{\psi_{X\times Y, Z}} & FX\times F(Y\times Z)\ar[d]_{\psi_{X,Y\times Z}}\\
F((X\times Y)\times Z)\ar[r]^{F\alpha_{X,Y,Z}} & F(X\times (Y\times Z)))
}}\label{diag:monoidal:assoc}\tag{MF. 3}
\end{equation}
\end{minipage}}
\end{tabular}

Moreover, a monoidal functor is called \emph{symmetric} if
\begin{equation}
\vcenter{
\xymatrix
{
FX\times FY\ar[r]^{\psi_{X,Y}}\ar[d]_{\swp_{FX,FY}} & F(X\times Y)\ar[d]^{F\swp_{X,Y}}\\
FY\times FX\ar[r]_{\psi_{Y,X}} & T(Y\times X)
}}\label{diag:monoidal:sym}\tag{SYM}
\end{equation}

\end{definition}

\begin{definition}[Monoidal monad]
A \emph{monoidal monad} $(T,\eta, \mu)$ is a monad whose underlying functor is monoidal, 
such that the associated natural transformation 
$\psi_{X,Y}\colon  TX\otimes TY\to T(X\otimes Y)$ satisfies $\psi^0=\eta_1$
and makes the following diagrams commute, stating
that the unit and multiplication are monoidal natural transformations. 

\hspace{-1cm}
\begin{tabular}{l l}
\begin{minipage}{6cm}
\begin{equation}
\vcenter{
\xymatrix@C+1pc
{
X\times Y\ar[r]^{\eta_X\times \eta_Y} \ar[dr]_{\eta_{X\times Y}}& TX\times TY\ar[d]^{\psi_{X,Y}}\\
& T(X\times Y)
}}\label{diag:monoidal:unit}\tag{MM.1}
\end{equation}
\end{minipage}
&
\begin{minipage}{6cm}
\begin{equation}\vcenter{
\xymatrix@C-.2pc
{
T^2X\times T^2 Y\ar[d]^{\mu_{X}\times \mu_{Y}} \ar[r]^{\psi_{TX,TY}}
    & T(TX\times TY) \ar[r]^{T\psi_{X,Y}}& TT(X\times Y)\ar[d]_{\mu_{X\times Y}}\\
TX\times TY\ar[rr]^{\psi_{X,Y}} & & T(X\times Y)
}}\label{diag:monoidal:multiplication}\tag{MM.2}
\end{equation}
\end{minipage}
\end{tabular}
\end{definition}

\section{Details of Section \ref{sec:preserve}}
\label{sec:appendix_preserve}
Section \ref{sec:preserve} focuses on the mechanics of the preservation of algebraic features; we give here more details on our constructs.

\subparagraph*{Lifting monads.} 
 First, the lifting $\widehat{T}$ can be neatly formalised using a distributive law $\lambda \colon \poly T \Rightarrow T\poly$, 
as in~\cite{sokolova2007generic}.
To define it, first note that for any $\sigma \in \Sigma$ 
we have the map 
$$
\xymatrix{
	(TX)^{\ari{\sigma}} \ar[r]^-{\psi^{\ari{\sigma}}}
		& TX^{\ari{\sigma}} \ar[r]^-{T\kappa_\sigma} 
        & \displaystyle T\coprod_{\sigma \in \Sigma} X^{\ari{\sigma}}
},
$$
where $\kappa_\sigma$ is the coproduct injection. The distributive law 
$\lambda$ is the cotupling of these maps:
$$
\lambda_X \ \equiv \  \Bigl(
\xymatrix@C=1.7cm
{\displaystyle
\poly T X = \coprod_{ \sigma \in\Sigma} (TX)^{\ari{\sigma}} \ar[r]^-{\left[T\kappa_\sigma \circ \psi^{\ari{\sigma}}\right]_{\sigma \in \Sigma}}  & T \poly X
} 
\Bigr)
\,.
$$
This yields a lifting
$\widehat{T} \colon \Alg[\Sigma] \rightarrow \Alg[\Sigma]$ defined by
$$
\widehat{T}(A,a) \ \equiv \  \bigl(
	\xymatrix{
		\poly TA \ar[r]^-{\lambda_A}
			& T \poly A \ar[r]^-{Ta}
			& TA
	}
\bigr)\,.
$$
Since $\lambda$ is a distributive law of functor over monad~\cite{sokolova2007generic}, i.e., it is compatible 
with the monad structure of $T$, $\widehat{T}$ is a lifting of the \emph{monad} $T$.

\begin{remark}
Equivalently, $T$ preserves $E$ if $\widehat{T}$ restricts to $\Alg[\Sigma,E]$. The
latter is monadic; let $S$ be a monad whose category $\EM(S)$ of Eilenberg--Moore algebras is isomorphic
to $\Alg[\Sigma,E]$. Since liftings of $T$ to $\EM(S)$ correspond to distributive laws~\cite{Johnstone:Adj-lif},
preservation of $E$ implies the existence of a distributive law of the monad $S$ over the monad $T$. 

The converse does not hold: the existence of a distributive law $\lambda  \colon  ST \Rightarrow TS$
does not necessarily mean that $T$ preserves $E$. The point is that $S$ might have several different presentations
by operations and equations, and the notion of preservation makes explicit use of the presentation.
Consider for example $T$ to be a non-affine, relevant monad, and $S = \di$ the distribution monad. Algebras for $\di$ are convex algebras, which can be presented by several equivalent theories. One theory $\mathbb{T}_1$ has signature made of binary convex combination symbols $\oplus_{\lambda}$ for $\lambda \in [0,1]$, and the equations include the \emph{projection} axioms $\oplus_0(x,y)=y$ and $\oplus_1(x,y)=x$ (see for instance \cite{bonchi2017power} for more details). These two laws are one-drop equations and cannot be preserved by $T$. But $\di$ is also presented by $\mathbb{T}_2$, whose signature only contains convex combination operators $\oplus_{\lambda}$ for $\lambda \in ]0,1[$. Therefore $\mathbb{T}_2$ has no projection axioms; all its equations are linear or dup. As explained in \cite{layers}, the relevance of $T$ is sufficient to preserve $\mathbb{T}_2$, hence there exists a distributive law $\di T \Rightarrow T \di$. Indeed, the failure to preserve one presentation of $\di$ does not mean that no distributive law can be found. 
\end{remark}

\subparagraph*{Decomposing equations.} 
To study individual equations and their preservation, it is useful to redefine the interpretation
of terms by decomposing their action on variables, in the way described in~\cite{layers}, 
which we briefly recall here. 


Let $V\equiv \{x_1, x_2, \ldots \}$ be a set of variables. 
For a term $t \in \Sigma^* V$, 
we write $\Arg(t)$ for the \emph{list of arguments} used in $t$, ordered as they appear in $t$ from left to right.
For instance~$\Arg(t)=[x_1,x_4,x_1,x_2]$
    for~$t \equiv (x_1 \cdot x_4)+(x_1 \cdot x_2)$.
    
Now, let $\algb = (A,a)$ be a $\Sigma$-algebra. 
First, we define tor a term $t \in \Sigma^*V $ the transformation $\prepare(t):A^{|V|} \to A^k$, where $k=|\Arg(t)|$, in the following way. For $\Arg(t)=[x_{i_1},x_{i_2},x_{i_3}, \dots x_{i_k}]$, we have  $\prepare(t) = \langle \pi_{i_1},\pi_{i_2},\pi_{i_3}, \dots \pi_{i_k} \rangle$.
This map rearranges the inputs to match the layout of variables in $t$. Second, we perform the algebraic operations described by the symbols in $t$ using their interpretation in $\algb$. To this end and for a term $t \in  \Sigma^*V$, we use the transformation $\evaluate (t) \colon A^{k} \to A$, defined inductively as
$\evaluate(x) 
		= \id_A$ for $x \in V$, 
		and 	
		$\evaluate(\sigma(t_1, \ldots, t_{i})) 
		= A^{k} \xrightarrow{\evaluate(t_1) \times \ldots \times \evaluate(t_i)} A^{i} \xrightarrow{\sigma_{\mathcal{A}}} A$
		for all $\sigma \in \Sigma$. 

 For example, consider $V=\{x,y\}$ and the term $(x\cdot y) \cdot x$. The first transformation arranges the variables to match the arguments of the
 given term: $\prepare(t)=\langle \pi_1, \pi_2, \pi_1\rangle$. The second one evaluates in the algebra: $\evaluate(t) = \cdot_{\mathcal{A}} \circ  ( \cdot_{\mathcal{A}} \times \id)$.

By composing $\prepare$ and $\evaluate$, we can characterise satisfaction of an equation as follows. 
\begin{lemma}[\cite{layers}]\label{lem:soundcompl}
For $\algb$ a $\Sigma$-algebra and $t_1,t_2$ $\Sigma$-terms such that $V=\Var(t_1)\cup \Var(t_2)$,
    we have:
$\evaluate(t_1) \circ \prepare(t_1) \ =\  \evaluate(t_2) \circ \prepare(t_2) 
    \quad \Leftrightarrow\quad  \algb\models t_1 = t_2$.
\end{lemma}


By this characterisation, we can study whether $T$ preserves $t_1=t_2$ by examining whether 
$\evaluate(t_1) \circ \prepare(t_1) = \evaluate(t_2) \circ \prepare(t_2)$
implies $\evaluate[\widehat{T}\algb](t_1) \circ \prepare[\widehat{T}\algb](t_1) = \evaluate[\widehat{T}\algb](t_2) \circ \prepare[\widehat{T}\algb](t_2)$. 
The relevant properties are summarised in the following Lemma. 

\begin{lemma} 
\label{lemma:preservediag}
Consider the following diagram.
\begin{equation}
    \hspace{-1em}\vcenter{\xymatrix@C-.7pc@R-1.3pc{
& & (TA)^{|V|} \ar[d]_{\te} \ar@/^2.0pc/[ddrr]^{\prepare[\widehat{T} \algb](t_2)} \ar@/_2.0pc/[ddll]_{\prepare[\widehat{T} \algb](t_1)}& & \\
& \ccld{R_1} & T(A^{|V|})\ar[ld]_{T\prepare(t_1)} \ar[rd]^{T\prepare(t_2)}& \ccld{R_2} &\\
(TA)^{k_1}\ar[r]^{\te} \ar@/_2.0pc/[ddrr]_{\evaluate[\widehat{T} \algb](t_1)}& T(A^{k_1})\ar[ddr]^[@!-50]{T\evaluate[\algb](t_1)} & \ccld{b} & T(A^{k_2})\ar[ddl]_[@!50]{T\evaluate[\algb](t_2)} & (TA)^{k_2}\ar[l]_{\te} \ar@/^2.0pc/[ddll]^{\evaluate[\widehat{T} \algb](t_2)}\\
 & \ccld{a} &  & \ccld{c} \\
& & TA & &\\
    }}
\label{preservediag}
\end{equation}
In this context, we have:
	\begin{enumerate}
		\item $\ccld{b}$ commutes if $\algb \models t_1=t_2$;
		\item the outer diagram commutes iff $\widehat{T} \algb \models t_1=t_2$;
		\item $\ccld{a}$ and $\ccld{c}$ commute;
		\item if $\ccld{R_1}$ and $\ccld{R_2}$ commute then $T$ preserves $t_1=t_2$.
	\end{enumerate}
\end{lemma}
\begin{proof}
	The first two items hold by Lemma~\ref{lem:soundcompl}. The third is shown in~\cite{layers},
	and the fourth follows from the other items. 
\end{proof}
In the above diagram,
$\ccld{R_1}$ and $\ccld{R_2}$  are called \emph{residual diagrams}~\cite{layers}.
They give a sufficient condition for preservation.

\section{Details of Section \ref{sec:affine}}
\label{sec:appendix_affine}

We present now the omitted contents of Section~\ref{sec:affine}. First, we prove two auxiliary lemmas that will be used
in our proofs. Recall the \emph{unitor} isomorphism $\rho_X\colon  X \times 1 \to X$.
We write $\rho_n$ for the composition 
$\rho_n = \rho \circ (\id \times !) \colon  X \times 1^n \to X$. 
\begin{lemma}
For a monoidal monad $T$ on a Cartesian monoidal category~$\cat$, 
the diagram
\begin{equation}
    \vcenter{\xymatrix@R-1pc@C+1pc{
T1\times 1^n  \ar[r]^{\id \times \eta_1^n}  \ar[d]_{\rho_n}
& T1\times (T1)^n  \ar[d]^{\te}\\
 T1 
    &T(1 \times 1^n) \ar[l]^{T \rho_n}}}\label{diag:monofunctorn}
\end{equation}%
commutes for all $n \geq 1$.
\label{lemma:monoidalrhon}%
\end{lemma}%

\begin{proof}%
For~$n=1$ the result is
    simply~\eqref{diag:monoidal:left_unit}.
Assuming we've shown it for~$n$,
    consider the following diagram.

\noindent
\xymatrix@R-.8pc{
T1 \times 1^n 
\ar@/_1.5pc/@{{}{ }{}}[rrr]^{\hspace{2cm}\ccld{a}}
\ar[rrr]^{\id \times (\eta_1)^n}
\ar[rrd]^{\id \times \eta_{1^n}}
\ar[rdd]^{\id \times !}
\ar[dddd]_{\rho_n}
\ar@{}[rrdd]|{\ccld{c}}
& & & 
(T1)^{n+1}
\ar[dl]_{\id \times \te_n}
\ar[dddd]^{\te_{n+1}}
\\
& \hspace{1cm} & T1 \times T(1^n)
\ar@/^1.3pc/[rddd]^{\te}
\ar[d]_{\id \times T!}
& \hspace{-1.5cm}\ccld{d}
\\
\hspace{1cm}\ccld{b} & T1 \times 1 
\ar[ddl]^{\rho}
\ar[r]^{\id \times \eta_1}
&  T1 \times T1
\ar@{}[rd]|{\ccld{f}}
\ar[d]^{\te}
\\
& \hspace{1cm} & T(1 \times 1)
\ar[dll]_{T\rho}
&
\\
T1 
\ar@{}[rruu]|{\ccld{e}}
\ar@/^1.3pc/@{{}{ }{}}[rrr]^{\hspace{2cm}\ccld{g}}
& &  &
T(1^{n+1})
\ar[ul]_{T(id \times !)}
\ar[lll]_{T \rho_n}
}

\noindent
\ccld{a} commutes by~\eqref{diag:monoidal:unit}.
\ccld{b} and \ccld{g} 
    follow from the definition of~$\rho_n$.
\ccld{c} and \ccld{f} commute by naturality
of $\te$, \ccld{d} by monoidality, and \ccld{e} by the
\eqref{diag:monoidal:left_unit} again.
\end{proof}

\begin{lemma}\label{lemma:rhofinal}
Let $T$ be a monoidal monad on a Cartesian monoidal category $\cat$.
    Let $m,n \in \mathbb{N}$, let $f$ be a morphism $1^m \to 1^n$. Then the following diagram commutes:
\begin{equation}
    \vcenter{\xymatrix@R-1pc{
T1 \times 1^m 
\ar[rr]^{\id \times f}
\ar[dr]_{\rho_m}
&& 
T1 \times 1^n
\\
& T1
\ar[ur]_{(\rho_n)^{-1}}
}}
\end{equation}
\end{lemma}

\begin{proof}
Consider the following commuting diagram.
\begin{equation}
    \vcenter{\xymatrix@R-.8pc{
T1 \times 1^m 
\ar[rr]^{\id \times f}
\ar@/_/[ddr]_{\rho_m}
\ar[dr]^{\id \times !}
&& 
T1 \times 1^n
\ar@/^/[ddl]^{\rho_n}
\ar[dl]_{\id \times !}
\\
& 
T1 \times 1
\ar[d]_{\rho}
\\
& T1
}}
\end{equation}
Thus~$ \rho_m = \rho_n \circ (\id \times f)$,
whence
    $(\rho_n)^{-1} \circ \rho_m = \id \times f$.
\end{proof}

We now introduce the main technical lemma, crucial to our results. It expresses the following idea:
as we study the preservation of a one-drop equation by a monad $T$ on the trivial algebra $1$, recall diagram~\eqref{preservediag}. The subdiagrams labelled \ccld{R_1} and \ccld{R_2} may or may not commute, but if we precompose them with a certain morphism $\alpha$, we ensure that one of them commutes and that the other characterises the affineness of $T$.


\begin{lemma}\label{lemma:alpha}
Let $t_1=t_2$ be a one-drop equation with $t_1, t_2 \in \Sigma^* V$ and $\Var(t_1) \cup \Var(t_2) = V$. 
Let $T \colon \Set \rightarrow \Set$ be a monoidal monad. 
Then there exists an object $B$ and a morphism $\alpha\colon  B \to (T1)^{|V|}$ such that:
\begin{enumerate}
\item The following diagram commutes:
\begin{equation}
    \vcenter{\xymatrix@R-1pc{
& B \ar[r]^{\alpha}  \ar[d]_{\alpha}  & (T1)^{|V|} \ar[r]^{\prepare[T 1](t_2)} &  (T1)^{k_2} \ar[d]^{\te}\\
&(T1)^{|V|} \ar[r]_{\te}  & T(1^{|V|}) \ar[r]_{T\prepare[1](t_2)}&   T(1^{k_2})
}}
\label{alphacom}
\end{equation}
\item If the following diagram commutes, then $T$ is affine:
\begin{equation}
    \vcenter{\xymatrix@R-1pc{
& B \ar[r]^{\alpha}  \ar[d]_{\alpha}  & (T1)^{|V|} \ar[r]^{\prepare[T 1](t_1)} &  (T1)^{k_1} \ar[d]^{\te}\\
&(T1)^{|V|} \ar[r]_{\te}  & T(1^{|V|}) \ar[r]_{T\prepare[1](t_1)}&   T(1^{k_1})
}}
\label{alphanotcom}
\end{equation}
\end{enumerate}
Or the other way around, by substituting $t_1$ and $t_2$.
\end{lemma}

\begin{proof}
Let $t_1=t_2$ be a one-drop equation. Without loss of generality, we can say that there exists a variable $x$ appears once in $t_2$ but not in $t_1$. Again without loss of generality, we can reorder our variables and assume that $x$ is the first variable in $V$. Let $|V|=n+1$. We define $B=T1\times 1^n$ and $\alpha= \id \times (\eta_1)^n\colon  T1\times 1^n \to (T1)^{n+1}$.

Note that because the first variable is dropped in $t_1$, $\prepare(t_1)$ can be decomposed as $\prepare(t_1)= \rho \circ(! \times \beta)$ where $\beta$ is a natural transformation carrying out the rearrangement of the remaining inputs. Similarly, because $x$ appears only once in $t_2$, $\prepare(t_2)$ can be decomposed as $\prepare(t_1)= (\beta_2) \circ(\id \times \beta_1)$. The first input is not modified,  $\beta_1$ is the natural transformation that may drop or duplicate some of the other inputs, and finally $\beta_2$ realises a permutation of all its inputs.

We start with point 1.
Note that $k_2=Arg(t_2)>0$ as $t_2$ features at least one variable by assumption. Let us subdivide the diagram as follows:

\xymatrix@C-.2pc@R-.3pc{
T1 \times 1^n 
\ar@/^0.2pc/@{{}{ }{}}[drr]^{\ccld{c}}
\ar@/^0.2pc/@{{}{ }{}}[dddrr]^{\ccld{b}}
\ar[r]^{\id \times (\eta_1)^n}
\ar[drr]_{\id \times \beta_1}
\ar[dd]_{\id \times (\eta_1)^n}
\ar[ddr]_{\rho_n}
& 
(T1)^{n+1}
\ar[r]^{\id\times \beta_1}
&
(T1)^{k_2}
\ar@/^2.2pc/[ddd]^{\te}
\ar[r]^{\beta_2}
&
(T1)^{k_2}
\ar@{}[ddd]|{\ccld{e}}
\ar@/^1.5pc/[ddd]^{\te}
\\
&&
T1\times 1^{k_2-1}
\ar[u]^{\id \times \eta_1^{k_2-1}}
\\
(T1)^{n+1}
\ar@{}[r]|{\ccld{a}}
\ar[d]_{\te}
&
T1
\ar@/_1.2pc/@{{}{ }{}}[rruu]_{\ccld{d}}
\ar[rd]_{T (\rho_{k_2-1})^{-1}}
\ar[ur]_{(\rho_{k_2-1})^{-1}}
&
\\
T(1^{n+1})
\ar@/^1.3pc/@{{}{ }{}}[rr]^{\ccld{f}}
\ar[ur]_{T \rho_n}
\ar[rr]_{T(\id \times \beta_1)}
&&
T(1^{k_2})
\ar[r]_{T \beta_2}
&
T(1^{k_2})
}

\ccld{a} and \ccld{d} commute by Lemma~\ref{lemma:monoidalrhon}. \ccld{b} and \ccld{f} commute by Lemma~\ref{lemma:rhofinal}. \ccld{c} commutes by naturality of $\beta_1$. Finally, \ccld{e} commutes because $T$ is a monoidal monad: as it is shown in \cite{layers}, since $\beta_2$ encodes a permutation this diagram corresponds to a residual diagram for a linear equation, which is known to commute. We've proven point 1.

Next up is point 2.
We assume commutation of the diagram \eqref{alphanotcom} and show the commutation of the following diagram. Note that the outer diagram amounts to $\eta_1 \circ ! = \id$ and characterises the affineness of $T$.
%

\xymatrix@R-.5pc{
T1 
\ar@/_2.9pc/@{{}{ }{}}[ddddd]^{\ccld{b}}
\ar@{}[rrrd]|{\ccld{a}}
\ar[d]^{{\rho_n}^{-1}}
\ar[rrr]^{!}
\ar@/_3.29pc/[ddddd]_{\id}
&&&
1
\ar@{}[ddddl]^{\ccld{f}}
\ar[ddddd]^{\eta_1}
\\
T1\times 1^n
\ar[r]^{! \times \id}
\ar[d]_{\id \times (\eta_1)^n}
&
1\times 1^n
\ar@{}[rdd]|{\ccld{d}}
\ar[d]_{\id \times (\eta_1)^n}
\ar[r]_{\rho}
&
1^n
\ar[d]_{\beta}
&
\\
T1 \times(T1)^n
\ar@{}[rdd]|{\ccld{c}}
\ar[dd]_{\te}
&
1 \times(T1)^n
\ar[d]^{\id \times \beta}
& 
1^{k_1}
\ar@/^0.8pc/@{{}{ }{}}[dd]^{\ccld{e}}
\ar[uur]_{!}
\ar@/^2.2pc/[dd]^{\eta_{(1^{k_1})}}
\ar[d]_{(\eta_1)^{k_1}}
\\
&
1 \times(T1)^{k_1}
\ar[r]_{\rho}
& (T1)^{k_1}
\ar[d]_{\te}
\\
T(1^{n+1})
\ar@{}[rrrd]|{\ccld{g}}
\ar[r]_{T(!\times \beta)}
\ar[d]^{T\rho_n}
&
T(1\times 1^{k_1})
\ar[r]_{T(\rho)}
&
T(1^{k_1})
\ar[dr]_{T!}
\\
T1
&&&
T1
\ar[lll]_{\id}
}

\ccld{a} and \ccld{g} commute by finality. \ccld{b} commutes by property \eqref{diag:monofunctorn}. \ccld{c} corresponds to our assumption \eqref{alphanotcom}. \ccld{d} commutes by naturality of $\rho$ and \ccld{f} by naturality of $\eta$, whereas \ccld{e} commutes by monoidal property \eqref{diag:monoidal:unit}.
\end{proof}

\begin{proof}[Proof of Lemma~\ref{lemma:affine1}]
$t_1=t_2$ trivially holds on $1$, therefore we have:
\begin{equation}
\evaluate[1](t_1) \circ \prepare[1](t_1) = \evaluate[1](t_2) \circ \prepare[1](t_2)
\label{eq:holdson1}
\end{equation}
Let $\alpha$ be the morphism given by Lemma~\ref{lemma:alpha}. The equation is preserved by $T$, hence it holds on $T1$. Then we have:

\begin{align*}
& T\evaluate[1](t_1) \circ  \te \circ \prepare[T1](t_1) \circ \alpha \\
& \qquad\  =\ \evaluate[T1](t_1) \circ \prepare[T1](t_1) \circ \alpha 
& \text{Lemma }~\ref{lemma:preservediag}\\
&\qquad\  =\ \evaluate[T1](t_2) \circ \prepare[T1](t_2) \circ \alpha
& t_1=t_2\text{ holds on }T1\\
& \qquad\ =\ T\evaluate[1](t_2) \circ \te  \circ \prepare[T1](t_2) \circ \alpha
& \text{Lemma }~\ref{lemma:preservediag}\\
& \qquad\ =\ T\evaluate[1](t_2) \circ T \prepare[1](t_2) \circ \te   \circ \alpha
& \eqref{alphacom}\\
& \qquad\ =\ T\evaluate[1](t_1) \circ T \prepare[1](t_1) \circ \te   \circ \alpha
& \eqref{eq:holdson1}\\
\end{align*}
The map $\evaluate[1](t_1)$ is an isomorphism, therefore we can precompose the previous equality with $T \evaluate[1](t_1)^{-1}$ and obtain \eqref{alphanotcom}, hence $T$ is affine by Lemma~\ref{lemma:alpha}.
%
%
\end{proof}

\begin{proof}[Proof of Theorem~\ref{thm:affine-undecidable}]
With respect to the proof in the main text, 
what remains to be shown is that $T1$ is isomorphic to $\mathcal{M}$. We write $e$ for the unit of $\mathcal{M}$. Each element of $T1$ can be seen as a $(\Sigma,E)$-term on generator $*$,
through $F \colon T1 \rightarrow \mathcal{M}$ defined by $F(x) = e$ and 
$F(f_{g_1}(f_{g_2}(\dots f_{g_n}(*)\dots )))  = g_1 g_2  \dots  g_n$.
$F$ is an isomorphism: by construction, $F(u)=F(v) \Leftrightarrow u=v$ for $u,v$ terms of $T1$, and all $m \in \mathcal{M}$ can be written as $F(t)$ for $t \in T1$.
Hence $T1=1$ iff $\mathcal{M}$ is trivial. This reduction gives that affineness is undecidable.
\end{proof}

\section{Details for Section \ref{sec:dup}}
\label{sec:appendix_dup}

The proofs of several results have been omitted in section~\ref{sec:dup}.

\begin{proof}[Proof of Lemma \ref{lemma:fgsleeve}] For $(i)$, remark that $(\id \times !)$ is an epimorphism, which gives immediately $f=g$ . For the case $(ii)$, we remark that $(\id \times T!)$ is an epimorphism and we have:
\begin{align*}
f \circ T(\id \times !) \circ \te_{X,1}  &=g  \circ T(\id \times !) \circ \te_{X,1}\\
f \circ \te_{X,1} \circ (\id \times T!) &=g  \circ \te_{X,1} \circ (\id \times T!) & \text{ naturality of }\te \\
f \circ \te_{X,1}&=g  \circ \te_{X,1}& \text{ epimorphism} \\
f \circ \te_{X,1} \circ (\id \times \eta_1) \circ \rho_X^{-1} \circ T \rho_X
&=
g  \circ \te_{X,1} \circ (\id \times \eta_1) \circ \rho_X^{-1} \circ T \rho_X \\
f&=g & \eqref{diag:monoidal:left_unit}
\end{align*}

%
%

\end{proof}

\begin{proof}[Proof of Lemma \ref{msatisfies}]
We show that the evaluations of $(yx)z$ and $(y(xx))z$ are  equal when variables $x,y,z$ are substituted with elements of $X^5$ (that we will also write $x,y,z$ for simplicity) and the beach binary operation is interpreted with $m$. In this framework, this means that on both sides of the equality \eqref{bigsleeve}, the wires reaching the top of the diagrams are the same.

We trace down each thread to explain the construction:
    the first wire is labelled $L^*$ because it connects the first component of the output to the first component of the left input.    Therefore, no matter how many occurrences of~$m$ are connected in whichever way, this thread
    always leads to the wire on the far left. In both cases $(yx)z$ and $(y(xx))z$, it will connect to the first component of $y$.
    Similarly the~$R^*$ wire goes through every occurrence of $m$ on the far right
        (in our case ending up in~$z$).
The~$LR^*$ wire turns left, then connects to the $R^*$ output of the next box: from that point, just like the previous thread, it always goes right. In the case of the left hand side term of \eqref{yxz}, it connects directly to~$x$ on its second component. In the term $(y(xx))z$, it goes through one more $m$-box and connects to the second component of $x_2$.
The next wire~$RL^*$ is connected to the $L^*$ thread of the right input, which for both sides of the equations connects to the first component of the variable~$z$.
The final wire~$LRL^*$ takes it a step further,
    going left then right and connecting to a~$L^*$ output. On the left hand side of the equation, it connects to the first component of~$x$. On the right hand side, it goes through one more instance of $m$ and connects to the first component of $x_1$.
Therefore the outputs of both sides of the equation match, in other words $m$ satisfies $x  (y  y) = yx$ .
\end{proof}

\begin{proof}[Proof of Proposition~\ref{prop:relevance-weakening}]
We proceed by induction on $n$, the case $n=2$ is trivial. We assume
that $T$ is relevant and $(n-1)$-relevant, we show that it is
$n$-relevant:
\begin{align*}
\te^n \circ \Delta^n 
    &\ =\ \te \circ (\id \times \te^{n-1}) \circ (\id \times \Delta^{n-1}) \circ  \Delta 
\ =\ \te \circ (\id \times T \Delta^{n-1}) \circ  \Delta \\
    & \ =\ (\id \times T \Delta^{n-1}) \circ \te \circ  \Delta 
\ =\ (\id \times T \Delta^{n-1}) \circ T \Delta 
     \ =\ T \Delta^n && \qedhere
\end{align*}
\end{proof}

\begin{proof}[Proof of Proposition~\ref{prop:relevant-and-affine}]
Note that affineness implies $\te \circ \langle \pi_1 , \pi_2 \rangle
= T \langle \pi_1 , \pi_2 \rangle \circ \te^n$ (see for instance
\cite{layers} as this equality corresponds to a residual diagram
for the equation $(x\cdot y)\cdot z=x\cdot y$)
and so
\begin{align*}
\te \circ \Delta 
\  =\  \te \circ \langle \pi_1 , \pi_2 \rangle \circ \Delta^n
\  =\  T \langle \pi_1 , \pi_2 \rangle \circ \te^n \circ \Delta^n 
\  =\  T \langle \pi_1 , \pi_2 \rangle \circ T \Delta^n 
    \  =\  T \Delta,
\end{align*}
where the final equality is due to the~$n$-relevance of~$T$.
\end{proof}

\begin{proof}[Proof of Theorem~\ref{thm:nrelevance}]
$(\Leftarrow)$:
Consider an algebra $\algb$ where  $f^n(x, \dots x)=x$ holds.
We refer to diagram \eqref{preservediag} and show that $\ccld{D_1}$ and $\ccld{D_2}$ commute. $\id \circ \prepare[\widehat{T} \algb](x) = T \prepare[\algb](x) \circ \id$ trivially holds because $\prepare(x)=\prepare[\widehat{T} \algb](x)=\id$. Furthermore, note that $\prepare(f^n(x, \dots x))= \prepare[\widehat{T} \algb](f^n(x, \dots x))= \Delta^n$, therefore by $n$-relevance we have: $\te^n \circ \prepare[\widehat{T} \algb](f^n(x, \dots x)) = T \prepare[\algb] (f^n(x, \dots x)) \circ \id$. By Lemma~\ref{lemma:preservediag}, the equation is preserved.

$(\Rightarrow)$: We assume that $f^n(x, \dots x)=x$ is preserved on every algebra. Then in particular, it is preserved for the $n$-ary operation $f^n= \pi_1 \times \dots \times \pi_n $ defined on a set $A^n$. The equation holds, in other words we have $(\pi_1 \times \dots \times \pi_n)  \circ \Delta^n = \id$. Then by the same reasoning as in Theorem~\ref{relevantiffdup}, we obtain:
\begin{align*}
    \id &\  =\  Tf^n \circ \te^n \circ \Delta^n
 \ =\  Tf^n \circ \te^n \circ \chi^n \circ T\Delta^n
 \ =\ T(\pi_1 \times \dots \times \pi_n) \circ \te^n \circ \chi^n \circ T\Delta^n\\
    &\ =\ \te^n  \circ (T\pi_1 \times \dots \times T\pi_n) \circ \chi^n \circ T\Delta^n
    \ =\ \te^n  \circ \chi^n \circ T(\pi_1 \times \dots \times \pi_n)  \circ T\Delta^n\\
    &\ =\  \te^n \circ \chi^n \circ T((\pi_1 \times \dots \times \pi_n)  \circ \Delta^n)
 \ =\  \te^n \circ \chi^n   \qedhere
\end{align*}
\end{proof}

\begin{proof}[Proof of Proposition~\ref{prop:multiset-preserves}]
Before we prove this, we note that $x  (y  y) = yx$
    implies~$xy=yx$:
\begin{equation*}
    xy \,=\, y  (x  x) 
        \,=\, (x  x)  (y  y)
        \,=\, (x  x)  (y  (y  y))
        \,=\, (x  x) ((yy)  (y  y))
        \,=\, (yy)  (x  x) 
        \,=\, x  (yy)
        \,=\, yx.
\end{equation*}
To show~$\mzt$ preserves $x  (y  y) = yx$,
    suppose~$m \colon X^2 \to X$ is given
        with~$m(x, m(y,y)) = m(y,x)$ for all~$x,y \in X$.
Note~$
    \overline{m}(\xi, \xi')
    =  \sum_{x,x'} \xi(x) \xi'(x') m(x,x')$
    for any~$\xi,\xi' \in \mzt X$,
    thus
    \begin{equation*}
        \overline{m}(\xi,\xi)
           \ =\  \sum_x \xi(x)^2 m(x,x)
           \ =\  \sum_x \xi(x) m(x,x)
    \end{equation*}
            as the off-diagonals cancel each other
                due to~$m(x,y) + m(y,x) = 2 m(x,y) = 0$.
                Hence
\begin{equation*}
\overline{m} (\xi, \overline{m} (\xi', \xi'))
      \ =\  \sum_{x,y \in X} \xi(x) \xi'(y) m(x, m(y,y)) 
      \ =\  \sum_{x,y \in X} \xi(x) \xi'(y) m(y, x) 
      \ =\  \overline{m}(\xi', \xi).
\end{equation*}
Thus we have shown that~$\mzt$
    preserves the 2-dup non-drop equation~$x(yy) = yx$.
\end{proof}

\begin{proof}[Proof of Proposition~\ref{propalgrel}]
Assume~\eqref{algebraicrelevance} holds.
    Let~$X$ be given.
    Any element of~$T (X\times X)$
    is of the form~$M[ \vec{x} \otimes \vec{y} ]$,
        where~$M$ is a~$\Th$-term
        and~$\vec{x}\otimes \vec{y} \equiv ( (x_1,y_1), \ldots,
                (x_n, y_n) )$ for~$x_i,y_i \in X$.
    Note that~$\chi(M[\vec{x} \otimes \vec{y}])
        = (M[\vec{x}], M[\vec{y}])$.
We will prove by induction over~$M$
    that
\begin{equation}\label{algebraicrelevanceih}
    \psi(M[\vec{x}], M[\vec{y}]) \ =\  M[\vec{x} \otimes \vec{y}],
\end{equation}
which shows that~$T$ is relevant.
So assume~$f$ is any~$n$-ary operation of~$\Th$
    and~$M_i$ are $\Th$-terms for which~\eqref{algebraicrelevanceih} holds.
We compute
\begin{equation*}
    \begin{split}
    \psi \bigl(\,
    \vec{f} (( M_i[\vec{x}] )_i )\,,\,
    \vec{f} (( M_i[\vec{y}] )_i ) \,\bigr) 
        & \ = \ 
    \mtr{f} \bigl(\, (\psi(M_i[\vec{x}], M_j[\vec{y}]) )_{ij} \,\bigr) \\
        &
     \ \overset{\mathclap{\eqref{algebraicrelevance}}}{=} \ 
        \vec{f} \bigl(\,(\psi(M_i[\vec{x}], M_i[\vec{y}]))_i\,\bigr) 
     \ \overset{\mathclap{\eqref{algebraicrelevanceih}}}{=} \ 
        \vec{f} \bigl(\,(\psi(M_i[\vec{x} \otimes \vec{y}]))_i\,\bigr)
\end{split}
\end{equation*}
    and so indeed~\eqref{algebraicrelevanceih} holds for all~$M$
        by induction and so~$T$ is relevant.
    The proof of the converse is straightforward.
\end{proof}

\begin{proof}[Proof of Theorem~\ref{thmdiscerning}]
We will prove the result for a finitary monad --- that is, a monad that
is presentable by an algebraic theory.  The argument for arbitrary monads,
which are presented by infinitary algebraic theories, is similar but heavier
on paper.

So assume~$T$ is a finitary monad presented by an algebraic theory~$\Th$,
    such that~$T$ preserves a~$2$-discerning
    equation~$t_1 =t_2$.
Suppose~$f$ is any~$n$-ary operation of~$\Th$.
Write~$m$ for the number of variables in~$t_1$
    aside from the duplicated one.
Let~$X$ be any set with~$n\times n$ matrix~$(x_{ij})_{ij}$
    over it.  We have to show that~\eqref{algebraicrelevance} holds.
    We will take a slight detour:
let~$Y$ denote the free model of~$t_1=t_2$ over~$\{
    y_1, \ldots, y_n, r_1, \ldots, r_m\}$.
Using preservation of~$t_1=t_2$, we see that
\begin{equation*}
    \begin{split}
    \mtr{f} (\, (s_2 [y_i, y_j, \vec{r} ])_{ij} \, ) 
        &\ = \ \overline{s_2} [ \vec{f}(\vec{y}), \vec{f}(\vec{y}), \vec{r} ] 
     \ = \ \overline{t_2} [ \vec{f}(\vec{y}), \vec{r} ] 
     \ = \ \overline{t_1} [ \vec{f}(\vec{y}), \vec{r} ]  \\
        & \ = \ \vec{f}( \, (t_1[y_i, \vec{r}])_i \,) 
     \ = \ \vec{f}( \, (t_2[y_i, \vec{r}])_i \,) 
     \ = \ \vec{f}( \, (s_2[y_i, y_i, \vec{r}])_i \,).
    \end{split}
\end{equation*}
As~$t_1=t_2$ is~2-discerning
    the terms~$s_2[y_i, y_j, \vec{r}]$ are distinct
        and so there exists a map~$h\colon Y \to X$
        such that~$h(s_2[y_i,y_j, \vec{r}]) = x_{ij}$.
    Hence
\begin{equation*}
\begin{split}
 \mtr{f} ((x_{ij})_{ij})
      &\ =\  \mtr{f} \bigl(\,(h( s_2[y_i,y_j, \vec{r}] ))_{ij}\,\bigr) 
      \ =\  (Th)\bigl(\mtr{f} ( \, (s_2[y_i,y_j, \vec{r}])_{ij}\, )\bigr)  \\
      &\ =\  (Th)\bigl(\vec{f} ( \, (s_2[y_i,y_i, \vec{r}])_i \, )\bigr) 
      \ =\  \vec{f} ( \, (h(s_2[y_i,y_i, \vec{r}]))_i \, ) 
      \ =\  \vec{f} (\, (x_{ii})_i)\, ),
\end{split}
\end{equation*}
and so~$T$ is relevant by Proposition~\ref{propalgrel}.
\end{proof}

%
%
%
%
%
%
%
%
%

\end{document}